\documentclass[11pt]{article}

\RequirePackage[colorlinks,citecolor=blue,urlcolor=blue]{hyperref}
\RequirePackage[numbers]{natbib}

\usepackage{amsmath,amssymb,amsthm,bm,euscript,bbm}
\usepackage{verbatim}
\usepackage{graphicx}
\usepackage{epstopdf}
\usepackage{color}
\usepackage{pstricks,pst-node,pst-plot}
\usepackage{rotating}
\usepackage{enumitem}

\newcommand{\R}{\mathbb{R}}

\newcommand{\mH}{\mathcal{H}}

\renewcommand{\P}{{\rm P}}
\newcommand{\E}{{\rm E}}

\newcommand{\FDR}{{\rm FDR}}
\newcommand{\Power}{{\rm Power}}

\newcommand{\BH}{{\rm BH}}
\newcommand{\argmax}{\operatornamewithlimits{argmax}}
\newcommand{\SNR}{{\rm SNR}}

\def\l{{\langle}}
\def\r{\rangle}

\def\La{\Lambda}
\def\ep{\varepsilon}

\theoremstyle{plain}
\newtheorem{thm}{Theorem}
\newtheorem{lemma}[thm]{Lemma}
\newtheorem{prop}[thm]{Proposition}
\newtheorem{remark}[thm]{Remark}

\theoremstyle{definition}
\newtheorem{alg}[thm]{Algorithm}
\newtheorem{example}[thm]{Example}

\newcommand{\figurepath}{.}

\begin{document}

\title {Multiple Testing of Local Maxima for Detection of Peaks in Random Fields}
\author{Dan Cheng$^*$ and  Armin Schwartzman\footnote{Partially supported by NIH grant R01-CA157528. The authors thank Alexander Egner of Georg-August-Universit\"at G\"ottingen for providing the nanoscopy data used in this paper.} \\
Department of Statistics, North Carolina State University}

\maketitle

\begin{abstract}
A topological multiple testing scheme is presented for detecting peaks in images under stationary ergodic Gaussian noise, where tests are performed at local maxima of the smoothed observed signals. The procedure generalizes the one-dimensional scheme of \citep{Schwartzman:2011} to Euclidean domains of arbitrary dimension. Two methods are developed according to two different ways of computing p-values: (i) using the exact distribution of the height of local maxima \citep{CS:2014}, available explicitly when the noise field is isotropic; (ii) using an approximation to the overshoot distribution of local maxima above a pre-threshold \citep{CS:2014}, applicable when the exact distribution is unknown, such as when the stationary noise field is non-isotropic. The algorithms, combined with the Benjamini-Hochberg procedure for thresholding p-values, provide asymptotic strong control of the False Discovery Rate (FDR) and power consistency, with specific rates, as the search space and signal strength get large. The optimal smoothing bandwidth and optimal pre-threshold are obtained to achieve maximum power. Simulations show that FDR levels are maintained in non-asymptotic conditions. The methods are illustrated in a nanoscopy image analysis problem of detecting fluorescent molecules against the image background.
\end{abstract}

{\bf Keywords:} Gaussian random field; kernel smoothing; image analysis; overshoot distribution; topological inference; false discovery rate.

\section{Introduction}

Detection of sparse localized signals embedded in smooth noise is a fundamental problem in image analysis, with applications in many scientific areas such as neuroimaging \citep{Worsley:1996a,Genovese:2002,Taylor:2007}, microscopy \citep{Egner:2007,Geisler:2007} and astronomy \citep{Brutti:2005}. The key issue is to find a threshold to determine significant regions. This paper treats the thresholding problem as a multiple testing problem where tests are performed at local maxima of the observed image, allowing error rates and detection power to be topologically defined in terms of detected spatial peaks, rather than pixels or voxels.

Now commonplace in neuroimaging, Keith Worsley pioneered the use of random field theory to approximate the null distribution of the global maximum of the observed image to control the family-wise error rate (FWER) of detected voxels \citep{Worsley:1996a,Worsley:2004,Taylor:2007}. On the other hand, initial attempts to control the false discovery rate (FDR), desirable for being less conservative, ignored the spatial structure in the data \citep{Genovese:2002}. Recognizing the need to make inferences about connected regions rather than voxels in imaging applications, multiple testing methods have since been developed for pre-defined regions \citep{Heller:2006,Heller:2007,Sun:2014} and for the harder problem of detecting unknown clusters \citep{Pacifico:2004,Pacifico:2007,Zhang:2009}. It has been argued, however, that localized signal regions often present themselves as peaks in the image intensity profile, inviting a more powerful analysis based on local maxima of the observed data as the features of interest \citep{Poline:1997,Chumbley:2009}.

\citet{Schwartzman:2011} formalized peak detection by introducing a multiple testing paradigm where local maxima of the smoothed data are tested for significance. That work, however, was limited to one-dimensional spatial and temporal domains because the distribution of the height of local maxima, a key ingredient for calculation of p-values, has historically been known in closed-form only for one-dimensional stationary Gaussian processes \citep{Cramer:1967}. Recently, \citet{CS:2014} have obtained exact expressions for the height distribution of isotropic Gaussian fields and an approximation to the overshoot distribution of local maxima of constant-variance Gaussian fields by applying techniques from random matrix theory \citep{Fyodorov04}. These crucial developments allow us in the current paper to extend the multiple testing method of \citep{Schwartzman:2011} to Euclidean domains of higher dimension.

Our general algorithm consists of the following steps:
\begin{enumerate}[label={(\arabic*)}]
\item {\em Kernel smoothing}: to increase SNR \citep{Worsley:1996b,Smith:2009}.
\item {\em Candidate peaks}: find local maxima of the smoothed field above a pre-threshold.
\item {\em P-values}: computed at each local maximum under the complete null hypothesis of no signal anywhere.
\item {\em Multiple testing}: apply a multiple testing procedure and declare as detected peaks those local maxima whose p-values are significant.
\end{enumerate}

The main conceptual difference with the algorithm of \citet{Schwartzman:2011}, in addition to being multi-dimensional, is the introduction of a height pre-threshold in step (2). Pre-thresholding is often used in neuroimaging \citep{Zhang:2009} to reduce the number of candidate peaks or regions. Considered formally here, it leads to two ways of applying the above algorithm. If the exact distribution of the height of local maxima for computing p-values in step (3) is known, such as for isotropic fields \citep{CS:2014}, it is shown here that it is best not to apply pre-thresholding at all. However, if the distribution is unknown, as is the case to date for non-isotropic fields, then pre-thresholding is still valuable in that it enables the use of an approximation of the overshoot distribution of local maxima instead \citep{CS:2014}. In step (4), for concreteness, we focus on the Benjamini-Hochberg (BH) procedure \citep{Benjamini:1995} for controlling FDR, although other procedures and error criteria could be used. The algorithm is illustrated by a toy example in Figure \ref{fig:simul}.

\begin{figure}[t]
\begin{center}
\begin{tabular}{ccc}
\includegraphics[trim=30 10 30 10,clip,width=1.5in]{\figurepath/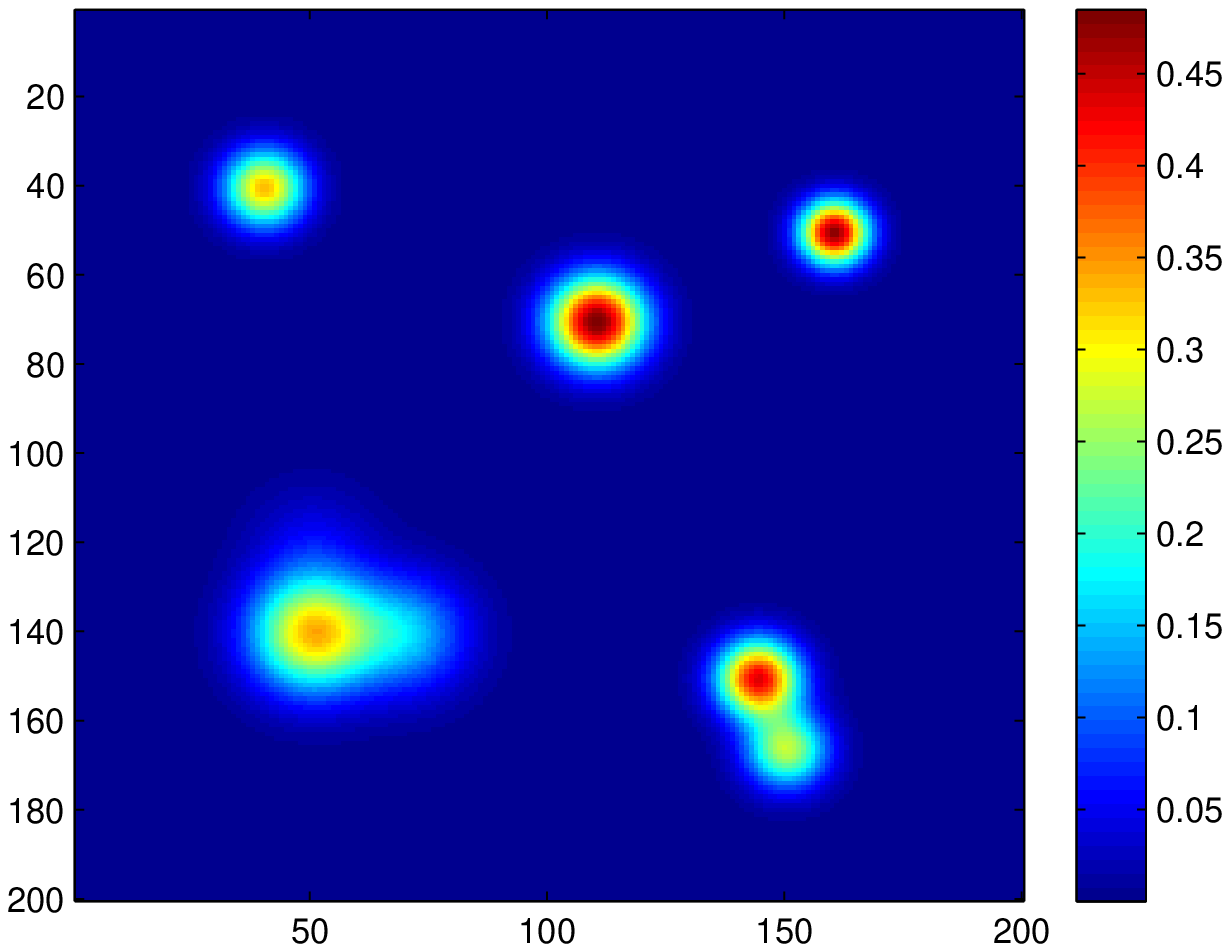} &
\includegraphics[trim=30 10 30 10,clip,width=1.5in]{\figurepath/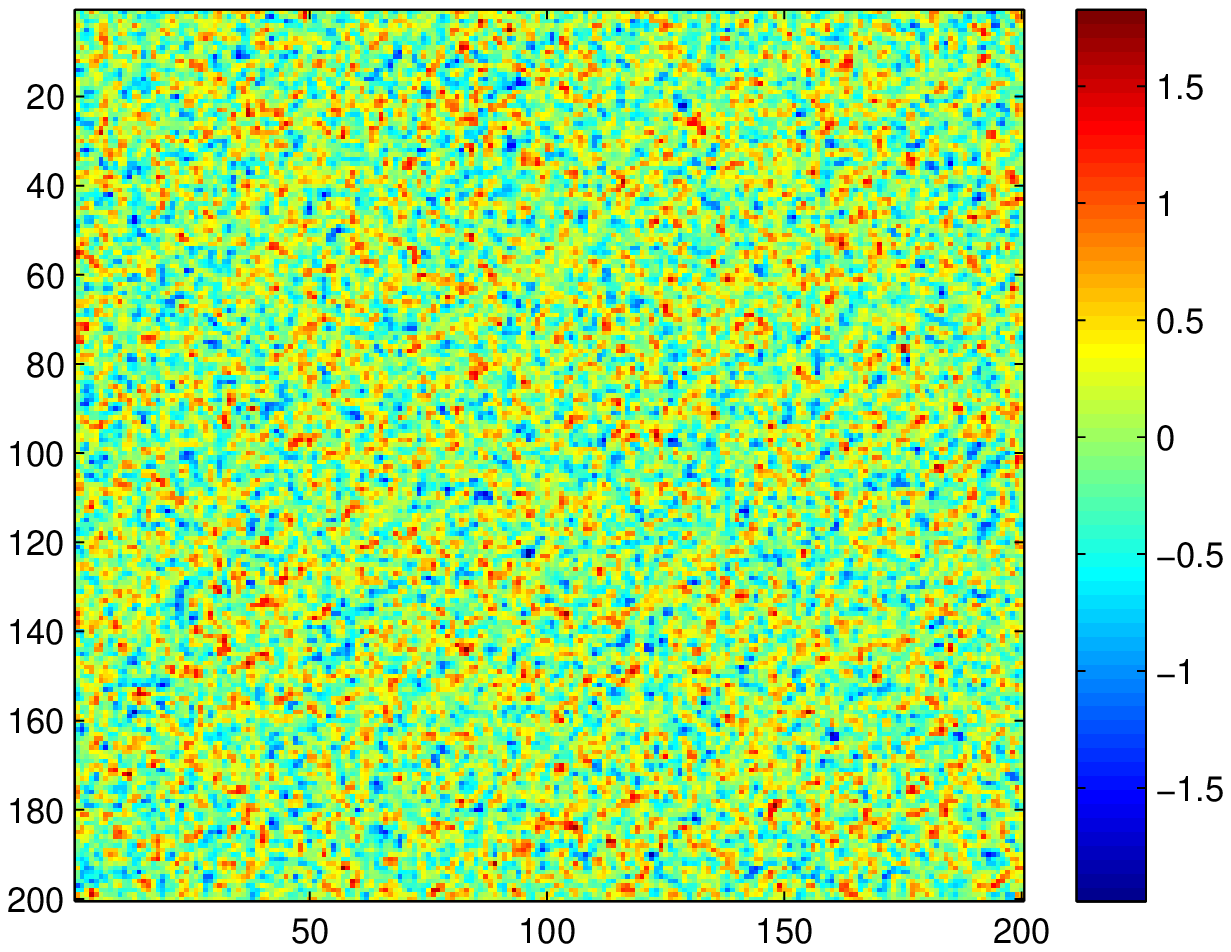} &
\includegraphics[trim=30 10 30 10,clip,width=1.5in]{\figurepath/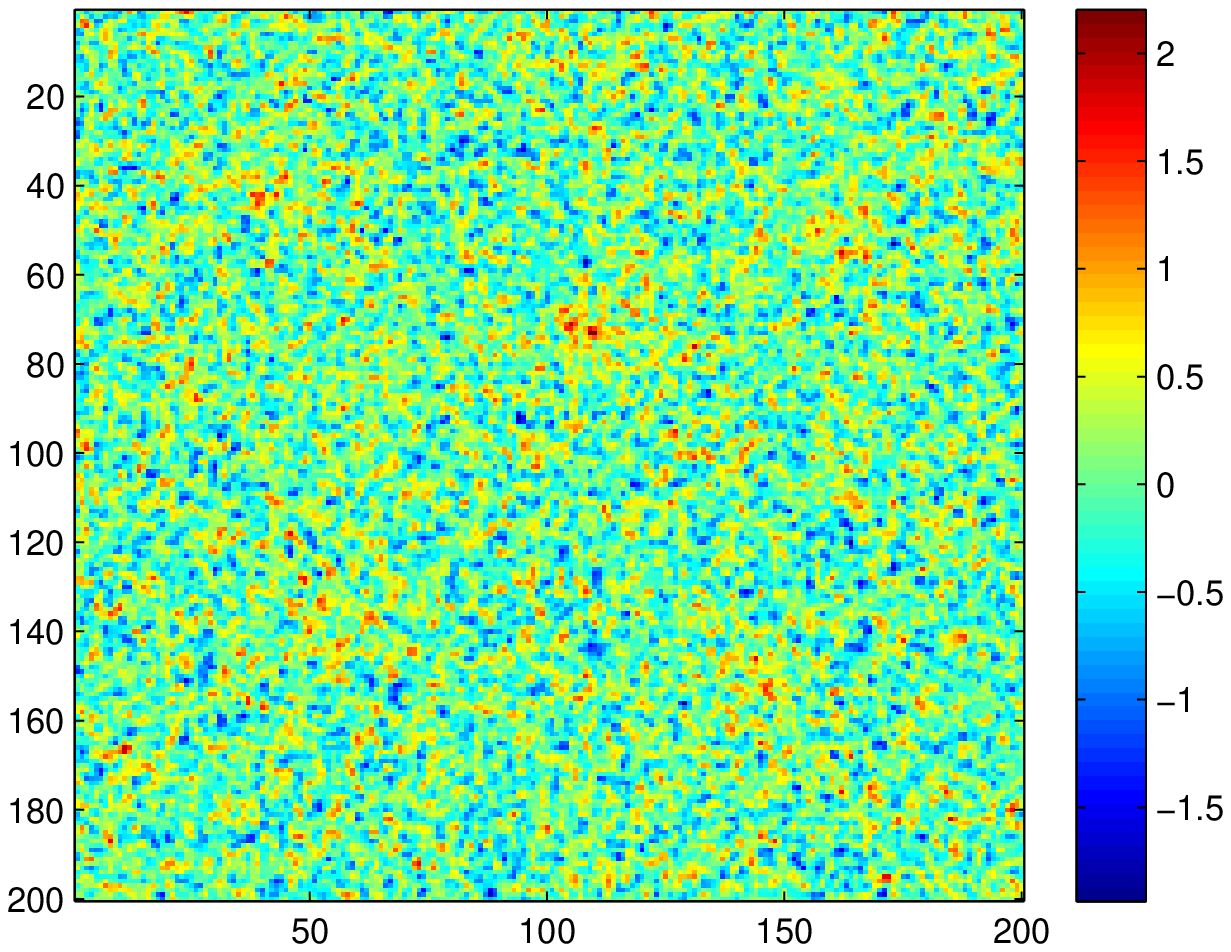} \\
\phantom{---} $\mu(t)$ (signal) & \phantom{---}$z(t)$ (noise) & \phantom{---}$y(t)=\mu(t)+z(t)$\\
\includegraphics[trim=30 10 30 10,clip,width=1.5in]{\figurepath/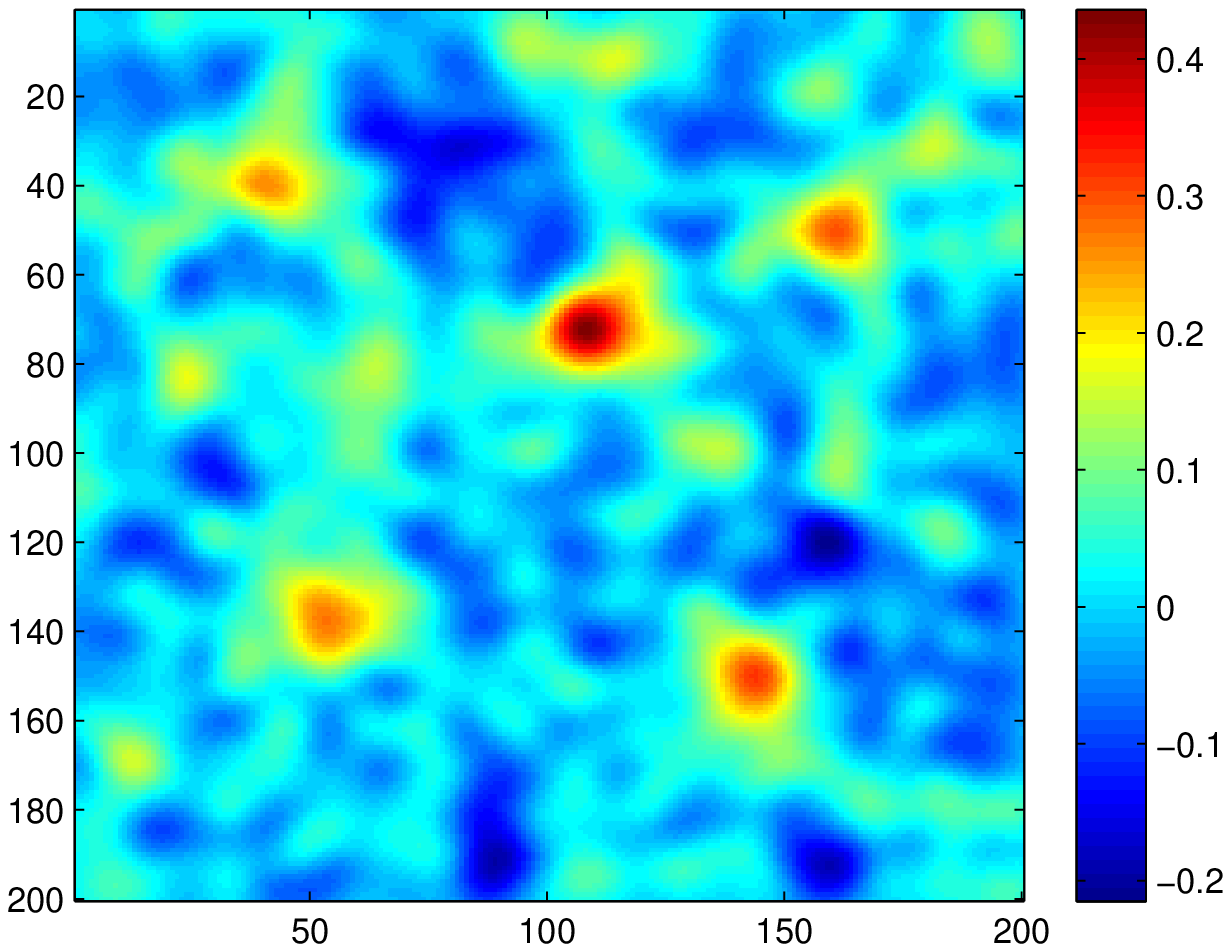} &
\includegraphics[trim=30 10 30 10,clip,width=1.5in]{\figurepath/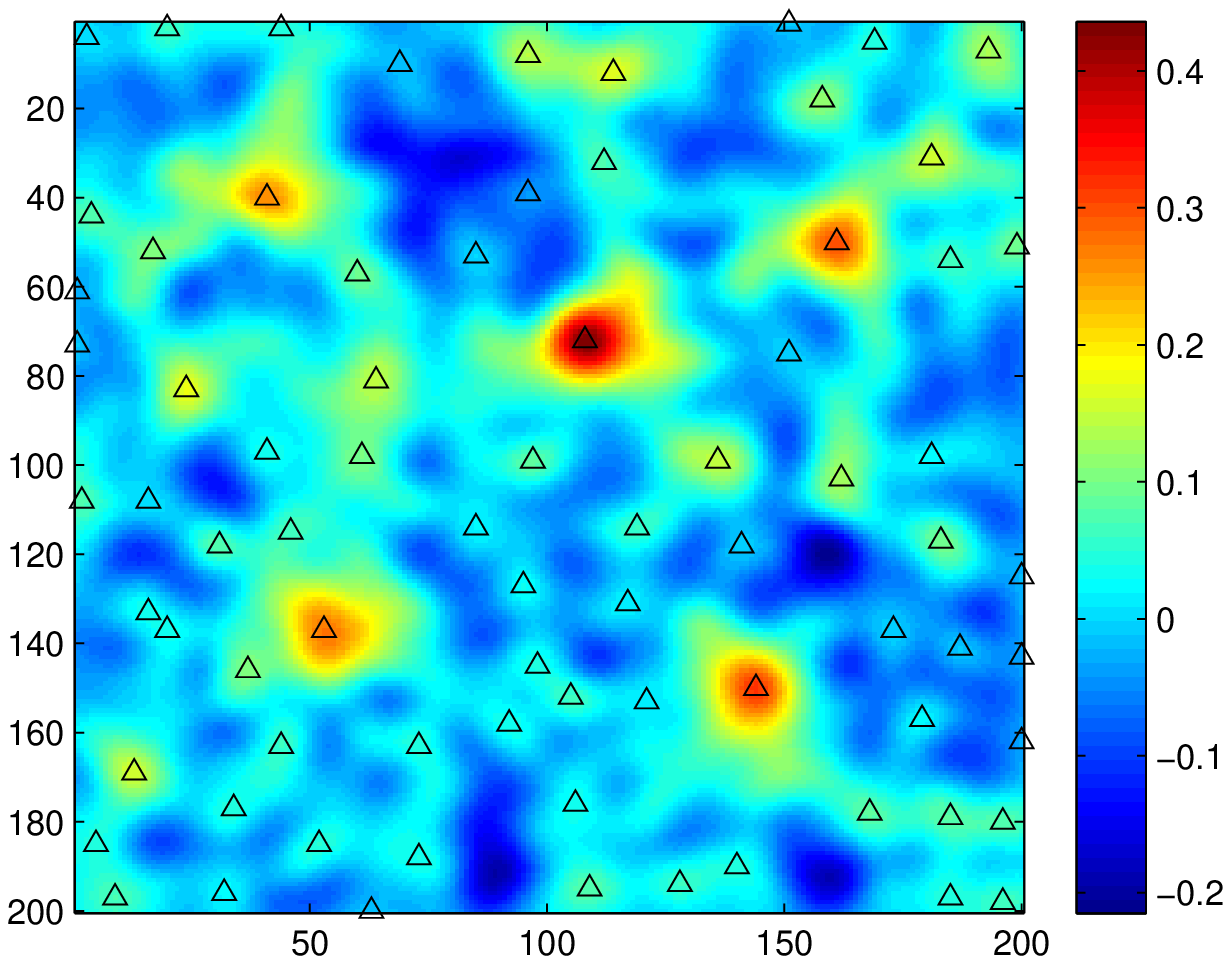} &
\includegraphics[trim=30 10 30 10,clip,width=1.5in]{\figurepath/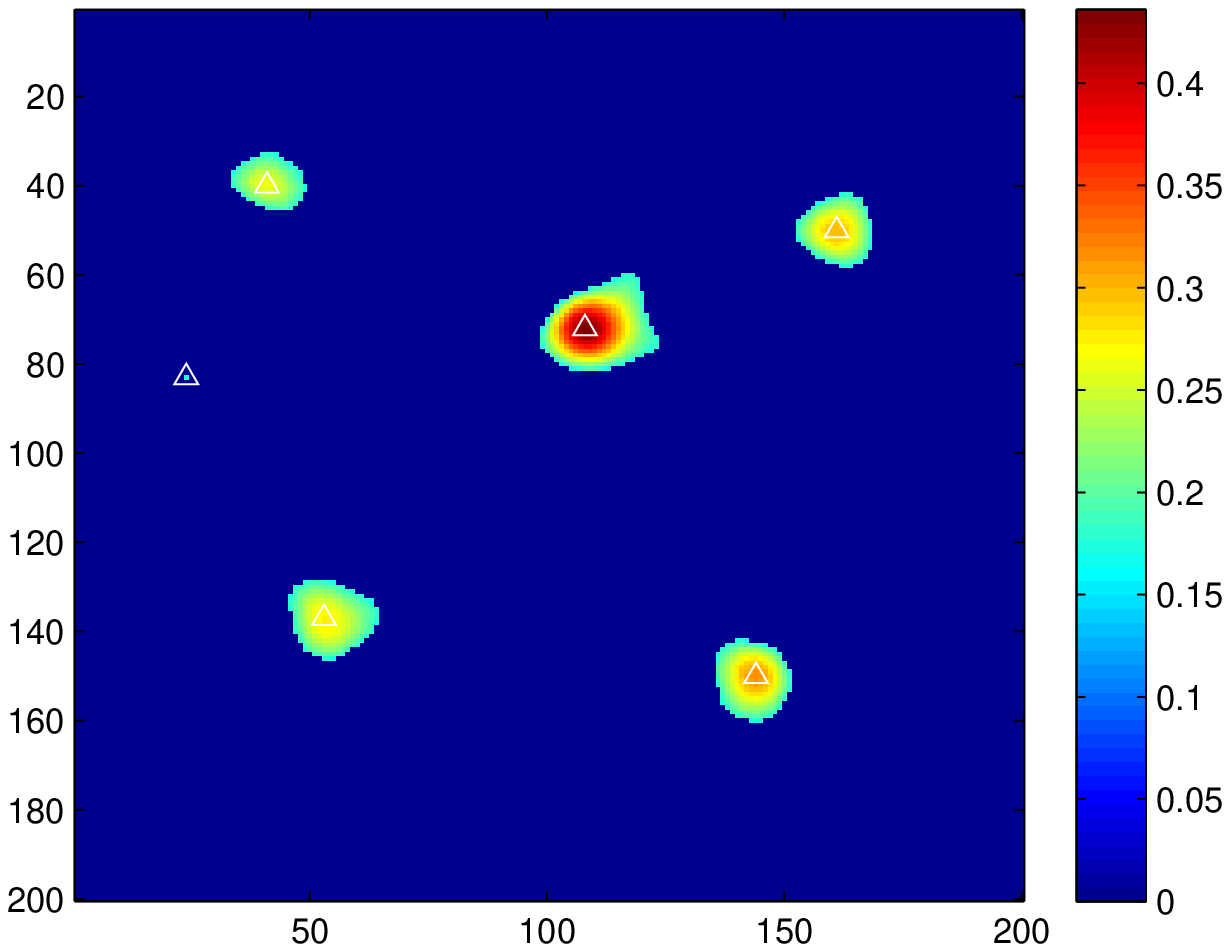} \\
\phantom{---}$y_\gamma(t)$ & \phantom{---}candidate peaks & \phantom{---}significant peaks
\end{tabular}
\caption{ \label{fig:simul} Raw signal $\mu(t)$ with six true peaks of different shapes and simulated Gaussian noise $z(t)$ produce the observed field $y(t)$ and smoothed field $y_\gamma(t)$. Out of 77 local maxima of $y_\gamma(t)$ (candidate peaks), the BH detection threshold at FDR level 0.2 selects six (significant peaks), one of which is a false positive. In this case, five out of six true peaks are detected.}
 \end{center}
 \end{figure}

Following the reasoning of \citet{Schwartzman:2011}, it is shown here that if the noise field is stationary and ergodic, then the proposed algorithm with the BH procedure provides asymptotic control of FDR and power consistency as both the search domain and the signal strength get large, the latter needing to grow only faster than the square root of the log of the former. The large domain assumption helps resolve an interesting aspect of inference for local maxima, namely the fact that the number of tests, equal to the number of observed local maxima, is random. The multiple testing literature usually assumes that the number of tests is fixed. The large domain assumption implies that, by ergodicity, the number of tests behaves asymptotically as its expectation. On the other hand, the strong signal assumption asymptotically eliminates the false positives caused by the smoothed signal spreading into the null regions, causing each signal peak region to be represented by only one observed local maximum within the true domain with probability tending to one. Simulations show that FDR levels are maintained and high power is achieved at finite search domains and moderate signal strength. We also find that the optimal smoothing kernel is approximately that which is closest in shape and bandwidth to the signal peaks to be detected, akin to the matched filter theorem in signal processing \citep{Pratt:1991,Simon:1995}. This bandwidth is much larger than the usual optimal bandwidth in nonparametric regression.

The results in this paper supercede those of \citet{Schwartzman:2011} in the sense that they can be seen as special cases when the domain is of one dimension and no pre-thresholding is applied. In addition, this paper provides specific rates for the asymptotic results, not available in \citet{Schwartzman:2011}, as well as a more rigorous discussion of the optimal smoothing bandwidth. Furthermore, the new concept of approximating p-values by pre-thresholding is not only useful in solving the multi-dimensional problem in this paper but it provides a potentially powerful tool for detection of peaks in non-stationary Gaussian fields on Euclidean space or manifolds \citep{CS:2014}.

The data analysis and all simulations were implemented in \texttt{Matlab}.


\section{The multiple testing scheme}
\label{sec:model-algorithm}


\subsection{The model}
\label{sec:model}
Consider the signal-plus-noise model
\begin{equation}
\label{eq:signal+noise}
y(t) = \mu(t) + z(t), \qquad t \in \R^N
\end{equation}
where the signal $\mu(t)$ is composed of unimodal positive peaks of the form
\begin{equation}
\label{eq:mu}
\mu(t) = \sum_{j=-\infty}^\infty a_j h_j(t), \qquad a_j > 0,
\end{equation}
and the peak shape $h_j(t) \ge 0$ has compact connected support $S_j = \{t: h_j(t) > 0\}$ and unit action $\int_{S_j} h_j(t)\,dt = 1$ for each $j$. Let $w_\gamma(t) \ge 0$ with bandwidth barameter $\gamma>0$ be a unimodal kernel with compact connected support and unit action. Convolving the process \eqref{eq:signal+noise} with the kernel $w_\gamma(t)$ results in the smoothed random field
\begin{equation}
\label{eq:conv}
y_\gamma(t) = w_\gamma(t) * y(t) =
\int_{\R^N} w_\gamma(t-s) y(s)\,ds = \mu_\gamma(t) + z_\gamma(t),
\end{equation}
where the smoothed signal and smoothed noise are defined as
\begin{equation}
\label{eq:mu-gamma}
\mu_\gamma(t) = w_\gamma(t) * \mu(t) = \sum_{j=-\infty}^\infty a_j h_{j,\gamma}(t), \qquad z_\gamma(t) = w_\gamma(t) * z(t).
\end{equation}

The smoothed noise $z_\gamma(t)$ defined by \eqref{eq:conv} and \eqref{eq:mu-gamma} is assumed to be a zero-mean thrice differentiable stationary Gaussian field such that for any non-negative integers $k_1, \ldots, k_N$ with $\sum_{i=1}^N k_i=k \in \{0, 1, 2, 3, 4\}$,
\begin{equation}
\label{Eq:variance-asymptotics}
\int_{\R_+^N} \left|\frac{\partial^k r_\gamma(t)}{\partial t_1^{k_1} \cdots \partial t_N^{k_N}} \right| dt<\infty,
\end{equation}
where $\R_+^N=[0,\infty)^N$ and $r_\gamma(t)=\E[z_\gamma(t)z_\gamma(0)]$. The technical condition \eqref{Eq:variance-asymptotics} is needed for obtaining the rates of FDR control and power consistency below, and by taking $k=0$, it implies the ergodicity of $z_\gamma(t)$ \citep{Cramer:1967}. It requires that the derivatives of the covariance function of the smoothed field $z_\gamma(t)$ should not decay too slowly. This can be easily obtained by using a Gaussian kernel $w_\gamma(t)$ in \eqref{eq:conv}, regardless of the smoothness of the original noise.

For each $j$, the smoothed peak shape $h_{j,\gamma}(t) = w_\gamma(t)*h_j(t) \ge 0$ is unimodal and has compact connected support $S_{j,\gamma}$ and unit action. For each $j$, we require that $h_{j,\gamma}(t)$ is twice differentiable in the interior of $S_{j,\gamma}$ and has no other critical points within its support. For simplicity, the theory requires that the supports $S_{j,\gamma}$ do not overlap although this is not crucial in practice.

Let $\tau_j$ be the unique point in the signal domain $S_{j,\gamma}$ where the peak shape $h_{j,\gamma}(t)$ attains its maximum. We impose additionally the following uniformity assumptions on the signal in our model.
\begin{enumerate}[label={(\arabic*)}]
\item $\sup_{j}|S_{j, \gamma}|<\infty$ and $\inf_{j} M_j>0$, where $M_j=h_{j,\gamma}(\tau_j)$.
\item There exists a universal $\delta>0$ such that $I_j^{\rm mode}:= \{t\in \R^N: \|t-\tau_j\| \le \delta\} \subset S_j$ for all $j$, $C=\inf_j C_j>0$ and $D=\inf_j D_j>0$, where
\begin{equation*}
\begin{split}
C_j &= \inf_{t\in I_j^{\rm side}} \l \nabla h_{j,\gamma}(t), (\tau_j-t)/\|\tau_j-t\|\r ,\quad I_j^{\rm side}=S_{j,\gamma}\backslash I_j^{\rm mode},\\
D_j &= -\sup_{t\in I_j^{\rm mode}} \sup_{\|x\|=1} x^T \nabla^2 h_{j,\gamma}(t) x.
\end{split}
\end{equation*}
\end{enumerate}
Here $\nabla f$ and $\nabla^2 f$ denote the gradient and Hessian of a real-valued function $f$, respectively.

Assumption (1) indicates that the sizes of the supports $S_{j, \gamma}$ are bounded and that the heights of the peaks of $h_{j,\gamma}$ are uniformly positive. Assumption (2) indicates that, uniformly for all $j$,  $h_{j,\gamma}(t)$ increases toward the mode $\tau_j$ along the direction $(\tau_j-t)/\|\tau_j-t\|$, and that $\sup_{\|x\|=1} x^T \nabla^2 h_{j,\gamma}(t) x$, which is the largest eigenvalue of matrix $\nabla^2 h_{j,\gamma}(t)$, is strictly negative in the vicinity of the mode so that the peak shape is strictly concave there.


\subsection{The STEM algorithm}
\label{sec:alg}
Suppose we observe $y(t)$ defined by \eqref{eq:signal+noise} in the cube of length $L$ centered at the origin, denoted by $U(L)=(-L/2,L/2)^N$, and suppose it contains $J$ peaks. We call the following procedure STEM (Smoothing and TEsting of Maxima).

\begin{alg}[STEM algorithm]
\label{alg:STEM}
\hfill\par\noindent
\begin{enumerate}[label={(\arabic*)}]
\item {\em Kernel smoothing}:
Construct the field \eqref{eq:conv}, ignoring the effects on the boundary of $U(L)$.
\item {\em Candidate peaks}:
For a fixed pre-threshold $v\in[-\infty,\infty)$, find the set of local maxima exceeding level $v$ for $y_\gamma(t)$ in $U(L)$
\begin{equation}
\label{eq:Tv}
\tilde{T}(v) = \left\{ t \in U(L): \ y_{\gamma}(t)>v, \
\nabla y_{\gamma}(t) = 0, \ \nabla^2 y_{\gamma}(t) \prec 0 \right\},
\end{equation}
where $\nabla^2 y_{\gamma}(t) \prec 0$ means that the Hessian matrix is negative definite.
\item {\em P-values}:
For each $t \in \tilde{T}(v)$, compute the p-value $p_\gamma(t,v)$ for testing the (conditional) hypothesis
$$
\mH_{0}(t): \ \mu(t) = 0 \quad \text{vs.} \quad \mH_{A}(t): \  \mu(t) > 0,
\qquad t \in \tilde{T}(v).
$$
\item {\em Multiple testing}:
Let $\tilde{m}(v) = \# \{t\in \tilde{T}(v)\}$ be the number of tested hypotheses. Apply a multiple testing procedure on the set of $\tilde{m}(v)$ p-values $\{p_\gamma(t,v), \, t \in \tilde{T}(v)\}$, and declare significant all local maxima whose p-values are smaller than the significance threshold.
\end{enumerate}
\end{alg}

When $v=-\infty$, we regard $\tilde{T}=\tilde{T}(-\infty)$ as the set of local maxima of $y_\gamma(t)$ in $U(L)$. In such case, Algorithm \ref{alg:STEM} becomes an $N$-dimensional version of the STEM algorithm proposed in \citep{Schwartzman:2011} for one-dimensional domains. When $v>-\infty$, an option not available in \citep{Schwartzman:2011}, Algorithm \ref{alg:STEM} provides a different way of selecting candidate peaks and computing p-values by choosing a pre-threshold $v$. In particular, this provides an efficient way to approximate the p-values for stationary and non-isotropic Gaussian noise (Section \ref{sec:overshoot}).

Steps (1) and (2) above are well defined under the model assumptions. Step (3) is detailed in Sections \ref{sec:height-distr} and \ref{sec:overshoot} below. For step (4), we use the BH procedure to control FDR (Section \ref{sec:FDR-control}). Notice that, in contrast to the usual BH procedures, the number of tests $\tilde{m}(v)$ is random.

\subsection{Error definitions}
\label{sec:errors}
As in \citep{Schwartzman:2011}, because the location of truly detected peaks may shift as a result of noise, we define a significant local maximum to be a true positive if it falls anywhere inside the support of a true peak. Conversely, we define it to be a false positive if it falls outside the support of any true peak.

Assuming the model of Section \ref{sec:model}, define the \emph{signal region} $\mathbb{S}_1 =  \cup_{j=1}^J S_j$ and \emph{null region} $\mathbb{S}_0 = U(L) \setminus \mathbb{S}_1$. For a significance threshold $u$ above the pre-threshold $v$, the total number of detected peaks and the number of falsely detected peaks are
\begin{equation}
\label{eq:R-and-V}
R(u) = \#\{t\in \tilde{T}(u)\} \quad {\rm and } \quad V(u) = \#\{t\in \tilde{T}(u)\cap \mathbb{S}_0\}
\end{equation}
respectively. Both are defined as zero if $\tilde{T}(u)$ is empty. The FDR is defined as the expected proportion of falsely detected peaks
\begin{equation}
\label{eq:FDR}
\FDR(u) = \E\left\{ \frac{V(u)}{R(u)\vee1} \right\}.
\end{equation}


Kernel smoothing enlarges the signal support and increases the probability of obtaining false positives in the null regions neighboring the signal \citep{Pacifico:2007}. Define the \emph{smoothed signal region} $\mathbb{S}_{1, \gamma} =  \cup_{j=1}^J S_{j, \gamma} \supset \mathbb{S}_1$ and \emph{smoothed null region} $\mathbb{S}_{0, \gamma} = U(L) \setminus \mathbb{S}_{1, \gamma} \subset \mathbb{S}_0$. We call the difference between the expanded signal support and the true signal support the \emph{transition region} $\mathbb{T}_{\gamma} = \mathbb{S}_{1, \gamma}\setminus \mathbb{S}_1 =\mathbb{S}_{0}\setminus \mathbb{S}_{0, \gamma} = \cup_{j=1}^J T_{j, \gamma}$, where $T_{j, \gamma} = S_{j, \gamma} \setminus S_j$ is the transition region corresponding to each peak $j$.

In general, more than one significant local maximum may be obtained within the domain of a true peak, affecting the interpretation of definition \eqref{eq:FDR}. However, this has no effect asymptotically because each true peak is represented by exactly one local maximum of the smoothed
observed field with probability tending to 1 (Lemma \ref{lemma:bounds} in Section \ref{app:lemmas}).


\subsection{Power}
\label{sec:power}
We define the power of Algorithm \ref{alg:STEM} as the expected fraction of true discovered peaks
\begin{equation}
\label{eq:power}
\Power(u) = \E \left( \frac{1}{J} \sum_{j=1}^{J} \mathbbm{1}_{\left\{
\tilde{T}(u) \cap S_j \ne \emptyset \right\}}\right) = \frac{1}{J} \sum_{j=1}^{J} \Power_j(u),
\end{equation}
where $\Power_j(u)$ is the probability of detecting peak $j$
\begin{equation}
\label{eq:power-j}
\Power_j(u) = \P\left(\tilde{T}(u) \cap S_j \ne \emptyset \right).
\end{equation}
The indicator function in \eqref{eq:power} ensures that only one significant local maximum is counted within the same peak support, so power is not inflated. Again, this has no effect asymptotically because each true peak is represented by exactly one local maximum of the smoothed observed process with probability tending to 1 (Lemma \ref{lemma:bounds} in Section \ref{app:lemmas}).


\newpage

\section{Detection of peaks by the height distribution of local maxima}
\label{sec:height-distr}


\subsection{P-values}
\label{sec:pvalue}
Given the observed heights $y_\gamma(t)$ at the local maxima $t\in\tilde{T}(v)$, the p-values in step (3) of Algorithm \ref{alg:STEM} are computed as
$p_\gamma(t,v) =  F_\gamma\left(y_\gamma(t),v\right)$, $t\in\tilde{T}(v)$, where
\begin{equation}
\label{eq:palm}
F_\gamma(u,v) = \P\left(z_\gamma(t) > u ~\Big|~ t \in \tilde{T}(v)\right)
\end{equation}
denotes the right tail probability of $z_\gamma(t)$ at the local maximum $t \in \tilde{T}(v)$, evaluated under the complete null hypothesis $\mu(t) = 0, \forall t$. By convention, when $v=-\infty$, denote
\begin{equation}\label{eq:F}
F_\gamma(u)=F_\gamma(u,-\infty).
\end{equation}

The conditional distribution \eqref{eq:palm} is a Palm distribution \citep[Ch. 6]{Adler:2010} and requires careful evaluation because the conditioning event has probability zero. Unlike the marginal distribution of $z_\gamma(t)$, it is not Gaussian but stochastically greater. Generally, for a constant-variance Gaussian field, there is an implicit formula for $F_\gamma(\cdot, v)$ \citep{CS:2014}. Theorem \ref{thm:Palm} below (\citep{CS:2014}, \citep{Adler:2010}) provides the formula for $F_\gamma(\cdot, v)$ for stationary Gaussian fields.

Let $\sigma_\gamma^2={\rm Var}(z_\gamma(t))$ and $\La_\gamma={\rm Cov}(\nabla z_\gamma(t))$, both independent of $t$ due to the stationarity of $z_\gamma(t)$. Denote by $\tilde{m}_{0,\gamma}(U(1))$ and $\tilde{m}_{0,\gamma}(U(1),u)$ respectively the number of local maxima of $z_\gamma(t)$ and the number of local maxima of $z_\gamma(t)$ exceeding level $u$ in the unit cube $U(1)=(-1/2,1/2)^N$. In fact, $\tilde{m}_{0,\gamma}(U(1))=\tilde{m}_{0,\gamma}(U(1),-\infty)$. By the Kac-Rice formula \citep{Adler:2007},
\begin{equation}\label{eq:expected-local-maxima}
\begin{split}
\E&[\tilde{m}_{0,\gamma}(U(1),u)] \\
&= \E\left[|{\rm det}\nabla^2 z_\gamma(t)| \mathbbm{1}_{\{z_\gamma(t)>u\}} \mathbbm{1}_{\{\nabla^2 z_\gamma(t)\prec 0\}}\big|\nabla z_\gamma(t)=0\right]p_{\nabla z_\gamma(t)}(0),
\end{split}
\end{equation}
where $p_{\nabla z_\gamma(t)}(0)=(2\pi)^{-N/2}({\rm det} (\La_\gamma))^{-1/2}$ is the density function of $\nabla z_\gamma(t)$ evaluated at $0$.

\begin{thm}
\label{thm:Palm}
Suppose the assumptions of Section \ref{sec:model} hold and that $\mu(t) = 0, \forall t$.
Then the distributions \eqref{eq:palm} and \eqref{eq:F} are given respectively by
\begin{equation*}
F_\gamma(u,v)= \frac{\E[\tilde{m}_{0,\gamma}(U(1),u)]}{\E[\tilde{m}_{0,\gamma}(U(1),v)]} \quad {\text and} \quad  F_\gamma(u)= \frac{\E[\tilde{m}_{0,\gamma}(U(1),u)]}{\E[\tilde{m}_{0,\gamma}(U(1))]}.
\end{equation*}
\end{thm}

It should be mentioned that the expectations above involving the indicator function $\mathbbm{1}_{\{\nabla^2 z_\gamma(t)\prec 0\}}$ are extremely hard to compute. Therefore, the explicit formula for $F_\gamma$ is usually unknown. The only exception, as far as we know, is the case when the field $z_\gamma$ is isotropic. This is because, in such case, one may apply the Gaussian Orthogonal Ensemble (GOE) technique from random matrix theory to compute these expectations \citep{Fyodorov04}. The corresponding explicit formula for $F_\gamma$ for isotropic Gaussian fields was recently obtained in \citep[Theorem 2.14]{CS:2014}. This will be used to compute the p-values exactly, see Proposition \ref{prop:Palm-isotropic} below.

\subsection{Error control and power consistency}
\label{sec:FDR-control}

Suppose the BH procedure is applied in step (4) of Algorithm \ref{alg:STEM}, as follows. For a fixed $\alpha \in (0,1)$, let $k$ be the largest index for which the $i$th smallest p-value is less than $i\alpha/\tilde{m}(v)$. Then the null hypothesis $\mH_0(t)$ at $t \in \tilde{T}(v)$ is rejected if
\begin{equation}
\label{eq:thresh-BH-random}
p_\gamma(t,v) < \frac{k\alpha}{\tilde{m}(v)}
\quad \iff \quad y_\gamma(t) >
\tilde{u}_{\BH}(v) = F_\gamma(\cdot, v)^{-1} \left(\frac{k\alpha}{\tilde{m}(v)}\right),
\end{equation}
where $k\alpha/\tilde{m}(v)$ is defined as 1 if $\tilde{m}(v)=0$. Since $\tilde{u}_{\BH}(v)$ is random, definition \eqref{eq:FDR} is hereby modified to
\begin{equation}
\label{eq:true-FDR}
\FDR_{\BH}(v) = \E\left\{ \frac{V(\tilde{u}_{\BH}(v))}{R(\tilde{u}_{\BH}(v))\vee1} \right\},
\end{equation}
where $R(\cdot)$ and $V(\cdot)$ are defined in \eqref{eq:R-and-V} and the expectation is taken over all possible realizations of the random threshold $\tilde{u}_{\BH}(v)$.

Define the conditions:
\begin{enumerate}
\item[(C1)] The assumptions of Section \ref{sec:model} hold.
\item[(C2)] $L \to \infty$ and $a = \inf_j a_j \to\infty$, such that $(\log L)/a^2 \to 0$, $J/L^N = A_1 + O(a^{-2}+L^{-N/2})$ and $|\mathbb{S}_{1,\gamma}|/L^N = A_{2,\gamma} + O(a^{-2}+L^{-N/2})$ with $A_1>0$ and $A_{2,\gamma} \in [0,1)$.
\end{enumerate}

\begin{thm}
\label{thm:FDR}
Let conditions (C1) and (C2) hold.

(i) Suppose that Algorithm \ref{alg:STEM} is applied with a fixed threshold $u>v$, then
\begin{equation}\label{eq:bound-FDR-u}
\FDR(u) \le \frac{\E[\tilde{m}_{0, \gamma}(U(1),u)](1-A_{2,\gamma})}{\E[\tilde{m}_{0, \gamma}(U(1),u)](1-A_{2,\gamma}) + A_1} + O(a^{-2}+L^{-N/2}).
\end{equation}

(ii) Suppose that Algorithm \ref{alg:STEM} is applied with the random threshold $\tilde{u}_{\BH}(v)$ \eqref{eq:thresh-BH-random}, then
\begin{equation}\label{eq:bound-FDR}
\FDR_{\BH}(v) \le \alpha\frac{\E[\tilde{m}_{0, \gamma}(U(1),v)](1-A_{2,\gamma})}{\E[\tilde{m}_{0, \gamma}(U(1),v)](1-A_{2,\gamma}) + A_1} + O(a^{-1}+L^{-N/4}).
\end{equation}
\end{thm}

The proof of Theorem \ref{thm:FDR} is given in Section \ref{app:FDR}. It can be seen from the proof of Theorem \ref{thm:FDR} that
the inequalities in \eqref{eq:bound-FDR-u} and \eqref{eq:bound-FDR} become equalities asymptotically (without specific rates), so the bounds given in \eqref{eq:bound-FDR-u} and \eqref{eq:bound-FDR} are tight and can be regarded respectively as the asymptotic estimators of $\FDR(u)$ and $\FDR_{\BH}(v)$. By \eqref{eq:u_BH^*}, the random threshold $\tilde{u}_{\BH}(v)$ converges asymptotically to the deterministic threshold
\begin{equation}
\label{eq:thresh-BH-fixed}
u^*_{\BH}(v) = F_\gamma^{-1}\left(\frac{\alpha A_1\E[\tilde{m}_{0,\gamma}(U(1), v)]/\E[\tilde{m}_{0,\gamma}(U(1))]}{A_1 + \E[\tilde{m}_{0,\gamma}(U(1),v)](1-A_{2,\gamma})(1-\alpha)}\right),
\end{equation}
which is a strictly increasing function in $v$. The threshold \eqref{eq:thresh-BH-random} can be viewed as the smallest solution of the equation $\alpha \tilde{G}(u,v) \ge F_\gamma(u,v)$, where $\tilde{G}(u,v)$ is the empirical right cumulative distribution function of $y_\gamma(t), ~t\in \tilde{T}(v)$ \citep{Genovese:2002}. Taking the limit of that equation as $L\to \infty$ yields the solution \eqref{eq:thresh-BH-fixed}.

Similar to the definition of $\FDR_{\BH}$ \eqref{eq:true-FDR}, since $\tilde{u}_{\BH}(v)$ is random, define
\begin{equation}
\label{eq:true-power}
\Power_{\BH}(v) = \E \left( \frac{1}{J} \sum_{j=1}^{J} \mathbbm{1}_{\left\{
\tilde{T}(\tilde{u}_{\BH}(v)) \cap S_j \ne \emptyset \right\}}\right).
\end{equation}
Since $\tilde{u}_{\BH}(v)$ converges to the deterministic threshold $u^*_{\BH}(v)$, which attains the minimum at $v=-\infty$, we see that the power is asymptotically maximized at $v=-\infty$ when $\gamma$ is fixed. This will be reflected in the simulation studies below (Figure \ref{fig:error-os}) and it shows that, if the exact height distribution of local maxima $F_\gamma(\cdot, v)$ or $F_\gamma(\cdot)$ is known, for example the smoothed noise $z_\gamma$ is an isotropic Gaussian field, then we will choose to apply the original STEM algorithm without pre-thresholding (i.e., $v=-\infty$) to perform the test.

The following lemma, whose proof is given in Section \ref{app:power}, provides an asymptotic approximation to the power at a fixed threshold.
\begin{lemma} \label{lemma:power-approx}
Let conditions (C1) and (C2) hold. As $a_j \to \infty$, the power for peak $j$ \eqref{eq:power-j} can be approximated by
\begin{equation}
\label{eq:approx-power}
\Power_j(u) = \Phi\left(\frac{a_j h_{j,\gamma}(\tau_j) - u}{\sigma_\gamma}\right)(1 + O(a_j^{-2})).
\end{equation}
\end{lemma}

The next result indicates that the BH procedure is asymptotically consistent. The proof is given in Section \ref{app:power}.
\begin{thm}
\label{thm:power}
Let conditions (C1) and (C2) hold.

(i) Suppose that Algorithm \ref{alg:STEM} is applied with a fixed threshold $u$, then
\[
\Power(u) = 1 - O(a^{-2}).
\]

(ii) Suppose that Algorithm \ref{alg:STEM} is applied with the random threshold $\tilde{u}_{\BH}(v)$ \eqref{eq:thresh-BH-random}, then
\[
\Power_{\BH}(v) = 1 - O(a^{-2}+L^{-N/2}).
\]
\end{thm}

The results in Theorem \ref{thm:power} rely on the condition $(\log L)/a^2 \to 0$ in (C2). If this condition is not satisfied, then it can be seen from the proofs of Lemmas \ref{lemma:power-approx} and \ref{lemma:unique-max} below that the power may be bounded above by a constant strictly less than one.


\subsection{Optimal smoothing kernel}
\label{sec:optimal-gamma}
The best smoothing kernel $w_\gamma(t)$ is that which maximizes the detection power under the true model. By Lemma \ref{lemma:power-approx}, the power \eqref{eq:approx-power} is approximately maximized by maximizing the signal-to-noise ratio (SNR)
\begin{equation}
\label{eq:SNR}
\SNR_{j,\gamma} = \frac{a_j h_{j,\gamma}(\tau_j)}{\sigma_\gamma} =
\frac{a_j \int_{\R^N} w_\gamma(s) h_j(s)\,ds}
{\sigma \sqrt{\int_{\R^N} w^2_\gamma(s)\,ds}},
\end{equation}
where $\sigma$ is the standard deviation of the observed process $y(t)$. The smoothing kernel $w_\gamma(t)$ that maximizes \eqref{eq:SNR} is called a matched filter in signal processing \citep{Pratt:1991,Simon:1995}. It is known in signal processing that if the peak locations are known, then the matched filter maximizes the detection power exactly. As shown in the simulations, the result only holds approximately in our case because the peak locations are unknown.


\subsection{Isotropic Gaussian fields}
\label{sec:p-value-isotropic}
The following result gives an explicit expression for the distribution \eqref{eq:palm} in the special case when $N=2$ and the noise field $z_\gamma(t)$ is isotropic. It is obtained from \citep[Example 2.16]{CS:2014} by standardizing the field in \eqref{eq:palm} as $F_\gamma(u,v) = \P(z_\gamma(t)/\sigma_\gamma > u/\sigma_\gamma | t \in \tilde{T}(v))$. Here, denote by $\phi(x)$ and $\Phi(x)$ respectively the density function and cumulative distribution function (cdf) of the standard Gaussian distribution.

\begin{prop}
\label{prop:Palm-isotropic}
Let the assumptions in Theorem \ref{thm:Palm} hold and let $z_\gamma(t)$ be an isotropic Gaussian field over $\R^2$ with correlation function
\begin{equation*}
\rho_\gamma(\|t-s\|^2)=\E[z_\gamma(t)z_\gamma(s)]/\sigma_\gamma^2
\end{equation*}
and $\rho_\gamma'' - \rho_\gamma'^2 \ge 0$, where $\rho_\gamma'=\rho_\gamma'(0)$ and $\rho_\gamma''=\rho_\gamma''(0)$. Let $\kappa_\gamma=-\rho_\gamma'/\sqrt{\rho_\gamma''}$, then $\E[\tilde{m}_{0,\gamma}(U(1))] = -\sqrt{3}\rho_\gamma''/(3\pi\rho_\gamma')$ and the distributions \eqref{eq:F} and \eqref{eq:palm} are given respectively by $F_\gamma(u) = \int_{u/\sigma_\gamma}^\infty g_\gamma(x) dx$ and $F_\gamma(u,v)=F_\gamma(u)/F_\gamma(v)$, where
\begin{equation*}
\begin{split}
g_\gamma(x) &=\sqrt{3}\kappa_\gamma^2(x^2-1)\phi(x)\Phi\left(\frac{\kappa_\gamma x}{\sqrt{2-\kappa_\gamma^2}} \right) + \frac{\kappa_\gamma x\sqrt{3(2-\kappa_\gamma^2)}}{2\pi}e^{-\frac{x^2}{2-\kappa_\gamma^2}} \\
&\quad+\frac{\sqrt{6}}{\sqrt{\pi(3-\kappa_\gamma^2)}}e^{-\frac{3x^2}{2(3-\kappa_\gamma^2)}}\Phi\left(\frac{\kappa_\gamma x}{\sqrt{(3-\kappa_\gamma^2)(2-\kappa_\gamma^2)}} \right).
\end{split}
\end{equation*}
\end{prop}

As shown in \citep{CS:2014},
$$
\rho_\gamma' = -\frac{1}{2\sigma_\gamma^2}{\rm Var}\left(\frac{\partial z_\gamma(t)}{\partial t_i}\right), \quad \rho_\gamma'' = \frac{1}{12\sigma_\gamma^2}{\rm Var}\left(\frac{\partial^2 z_\gamma(t)}{\partial t_i^2}\right)
$$
for any $i=1, \ldots, N$. Therefore, in order to estimate $\rho_\gamma'$ and $\rho_\gamma''$, we only need to estimate the variances of the derivatives of $z_\gamma$ (or equivalently $y_\gamma$).

\begin{example}[\bf{Gaussian autocorrelation model}]
\label{ex:Gaussian}
Let the noise $z(t)$ in \eqref{eq:signal+noise} be constructed as
$$
z(t) = \sigma \int_{\R^N} \frac{1}{\nu^N} \phi_N\left(\frac{t-s}{\nu}\right)\,dB(s), \qquad \sigma, \nu > 0,
$$
where $\phi_N(x) = (2\pi)^{-N/2}e^{-\|x\|^2/2}$ for all $x \in \R^N$ is the $N$-dimensional standard Gaussian density, $dB(s)$ is Gaussian white noise and $\nu > 0$ ($z(t)$ is regarded by convention as Gaussian white noise when $\nu=0$). Convolving with a Gaussian kernel $w_\gamma(t) = (1/\gamma^N)\phi_N(t/\gamma)$ with $\gamma > 0$ as in \eqref{eq:mu-gamma} produces a zero-mean infinitely differentiable stationary ergodic Gaussian field
$$
z_\gamma(t) = w_\gamma(t) * z(t) = \sigma \int_{\R^N} \frac{1}{\xi^N} \phi_N\left(\frac{t-s}{\xi}\right)\,dB(s), \qquad \xi = \sqrt{\gamma^2 + \nu^2},
$$
with $\sigma_\gamma^2 = \sigma^2/(2^N\pi^{N/2} \xi^N)$, $\rho_\gamma'=-(2\xi)^{-2}$, $\rho_\gamma''=(2\xi)^{-4}$ and $\kappa_\gamma=1$.
The above expressions may be used as approximations if the kernel, required to have finite support, is truncated at $[-\gamma d, \gamma d]^N$ for moderately large $d$, say $d=3$.

Suppose the signal peak $j$ is a truncated Gaussian density
$$
h_j(t) = (1/b_j^N)\phi_N[(t-\tau_j)/b_j]\mathbbm{1}_{\{(t-\tau_j)/b_j \in [-c_j,c_j]^N\}},\quad  b_j, c_j>0.
$$
Ignoring the truncation, $h_{j,\gamma}(t) = w_\gamma(t) * h_j(t)$ in \eqref{eq:SNR} is the convolution of two Gaussian densities with variances $\gamma^2$ and $b_j^2$, which is another Gaussian density with variance $\gamma^2 + b_j^2$. We have that
\begin{equation*}
\SNR_{j,\gamma} = \frac{a_j h_{j,\gamma}(\tau_j)}{\sigma_\gamma} = \frac{a_j}{\sigma\pi^{N/4}}\left[\frac{\gamma^2 + \nu^2}{(\gamma^2 + b_j^2)^2}\right]^{N/4}.
\end{equation*}
As a function of $\gamma$, the SNR is maximized at
\begin{equation}
\label{eq:optimal-gamma}
\argmax_\gamma \SNR_{j,\gamma} = \begin{cases}
\sqrt{b_j^2 - 2 \nu^2}, & \nu < b_j/\sqrt{2} \\
0, & \nu > b_j/\sqrt{2}.
\end{cases}
\end{equation}
In particular, when $\nu=0$, we have that the optimal bandwidth for peak $j$ is $\gamma = b_j$, the same as the signal bandwidth. We show in the simulations below that the optimal $\gamma$ is indeed close to \eqref{eq:optimal-gamma}. It can be seen from \eqref{eq:optimal-gamma} that as $\nu$ gets larger, which means that $y(t)$ gets smoother, the optimal $\gamma$ becomes smaller. If $\nu$ is large enough, there is no need to smooth at all.
\end{example}


\section{Detection of peaks by approximated overshoot distribution}
\label{sec:overshoot}


\subsection{Approximating the overshoot distribution}
In the neuroimaging literature, it has been proposed to pre-threshold the test statistic field and then perform the inference on the supra-threshold statistics \citep{Zhang:2009}. We showed theoretically in Section \ref{sec:FDR-control} and will confirm by simulations that, in the best case scenario where the exact distribution of the height of local maxima is known, pre-thresholding ($v=-\infty$) does not increase detection power. However, pre-thresholding is still very valuable if the exact distribution is unknown but an approximation is known.

As mentioned, if the Gaussian field is only stationary but not isotropic, then the explicit formula for $F_\gamma(u, v)$ \eqref{eq:palm} is unknown so far. By \citep[Corollary 2.7]{CS:2014}, there exists $\ep_0>0$ such that as $v \to \infty$ and $u>v$,
\begin{equation*}
F_\gamma(u,v) = K_\gamma(u,v)(1 + o(e^{-\ep_0 v^2})),
\end{equation*}
where
\begin{equation*}
K_\gamma(u,v) =\frac{H_{N-1}(u/\sigma_\gamma)e^{-u^2/(2\sigma_\gamma^2)}}{H_{N-1}(v/\sigma_\gamma)e^{-v^2/(2\sigma_\gamma^2)}}
\end{equation*}
and $H_{N-1}(x)$ is the Hermite polynomial of order $N-1$. A similar argument to the proof of \citep[Corollary 2.7]{CS:2014} yields that for a fixed $v$, as $u\to \infty$,
\begin{equation}\label{Eq:stationary-os-u}
F_\gamma(u,v)=\beta_\gamma(v)K_\gamma(u,v) (1+ o(e^{-\ep_0 u^2})),
\end{equation}
where
\begin{equation}\label{Eq:beta}
\beta_\gamma(v) = \frac{(2\pi)^{-(N+1)/2}\sigma_\gamma^{-N} ({\rm det(\La_\gamma)})^{1/2}H_{N-1}(v/\sigma_\gamma)e^{-v^2/(2\sigma_\gamma^2)}}{\E[\tilde{m}_{0, \gamma}(U(1),v)]}
\end{equation}
and $\La_\gamma={\rm Cov}(\nabla z_\gamma(t))$. Note that $\beta_\gamma(v)$ is similar to the ratio of the expected Euler characteristic \citep[Lemma 11.7.1]{Adler:2007} and the expected number of local maxima of $z_\gamma(t)$ over the unit cube $U(1)$. It is conjectured that $\beta_\gamma(v)<1$ for all $v>0$ (this is true for $N=1$ and $N=2$ \citep{CS:2014}).


\begin{thm}
\label{thm:FDR-os-approx}
Let conditions (C1) and (C2) hold. Suppose that Algorithm \ref{alg:STEM} is applied with the random threshold $\tilde{u}_{\BH}(v)$ by using $K_\gamma(u,v)$ instead of $F_\gamma(u,v)$ to compute p-values, then as $v\to \infty$ such that $v^2/\log(L)\to 0$ and $v^2/\log(a)\to 0$,
\begin{equation}\label{eq:bound-FDR-os-approx}
\FDR_{\BH}(v) \le \alpha\frac{\E[\tilde{m}_{0, \gamma}(U(1),v)](1-A_{2,\gamma})\beta_\gamma(v)}{\E[\tilde{m}_{0, \gamma}(U(1),v)](1-A_{2,\gamma}) + A_1} (1+ o(e^{-\ep_0 v^2})),
\end{equation}
where $\ep_0>0$ is some constant and $\beta_\gamma(v)$ is defined in \eqref{Eq:beta}, and moreover,
\begin{equation}\label{eq:power-os-approx}
\Power_{\BH}(v) = 1 - O(a^{-2}+L^{-N/2}).
\end{equation}
\end{thm}

Note that for a fixed threshold $u$, the control of $\FDR(u)$ and the consistency of power $\Power(u)$ in Theorem \ref{thm:FDR-os-approx} will be the same as those given in part (i) of Theorem \ref{thm:FDR} and part (i) in Theorem \ref{thm:power} respectively. From the proof of Theorem \ref{thm:FDR-os-approx}, we see that the limit is in fact taken when $u\to \infty$ while $v$ is fixed. However, we cannot tell the exact gap between $u$ and $v$, though it is assumed that $u$ is always greater than $v$. It is likely that $\tilde{u}_{\BH}(v)$ or $u^{**}_{\BH}(v)$ \eqref{eq:thresh-BH-fixed-os-approx} below is still relatively large for small $v$, which is true when applying the STEM algorithm by using $F(\cdot,v)$ to compute p-values. Therefore, in our simulations below (Figure \ref{fig:error-os-approx}), the approximation in Theorem \ref{thm:FDR-os-approx} even works well for small $v$.

\subsection{Optimal pre-threshold}
\label{sec:opt-v}
Under the conditions in Theorem \ref{thm:FDR-os-approx}, by \eqref{eq:u_BH^*-os-approx}, we see that the random threshold $\tilde{u}_{\BH}(v)$ converges asymptotically to the deterministic threshold
\begin{equation}
\label{eq:thresh-BH-fixed-os-approx}
u^{**}_{\BH}(v) = F_\gamma^{-1}\left(\frac{\alpha A_1\beta_\gamma(v)\E[\tilde{m}_{0, \gamma}(U(1),v)]/\E[\tilde{m}_{0, \gamma}(U(1))]}{A_1 + \E[\tilde{m}_{0,\gamma}(U(1),v)](1-A_{2,\gamma})(1-\alpha\beta_\gamma(v))}\right)(1+o(1)).
\end{equation}
For fixed $\gamma$, the power \eqref{eq:approx-power} is maximized at the optimal pre-threshold minimizing $u^{**}_{\BH}(v)$, which is
\begin{equation}\label{Eq:vopt}
v_{\rm opt}=\underset{v}{\operatorname{argmax}}
 \frac{H_{N-1}(v/\sigma_\gamma)e^{-v^2/(2\sigma_\gamma^2)}}{A_1 + \E[\tilde{m}_{0,\gamma}(U(1),v)](1-A_{2,\gamma})(1-\alpha\beta_\gamma(v))}.
\end{equation}
Let $\gamma$ and $\alpha$ be fixed, we see that $v_{\rm opt}$ depends only on the covariance structure of $z_\gamma(t)$.

\section{Simulation studies}
\label{sec:simulations}


\begin{figure}[t]
\begin{center}
\begin{tabular}{cccc}
& $\nu=0$ & $\nu=1$ & $\nu=2$ \\
\begin{sideways}\phantom{-----------}FDR\end{sideways} &
\includegraphics[trim=40 10 60 10,clip,width=1.4in]{\figurepath/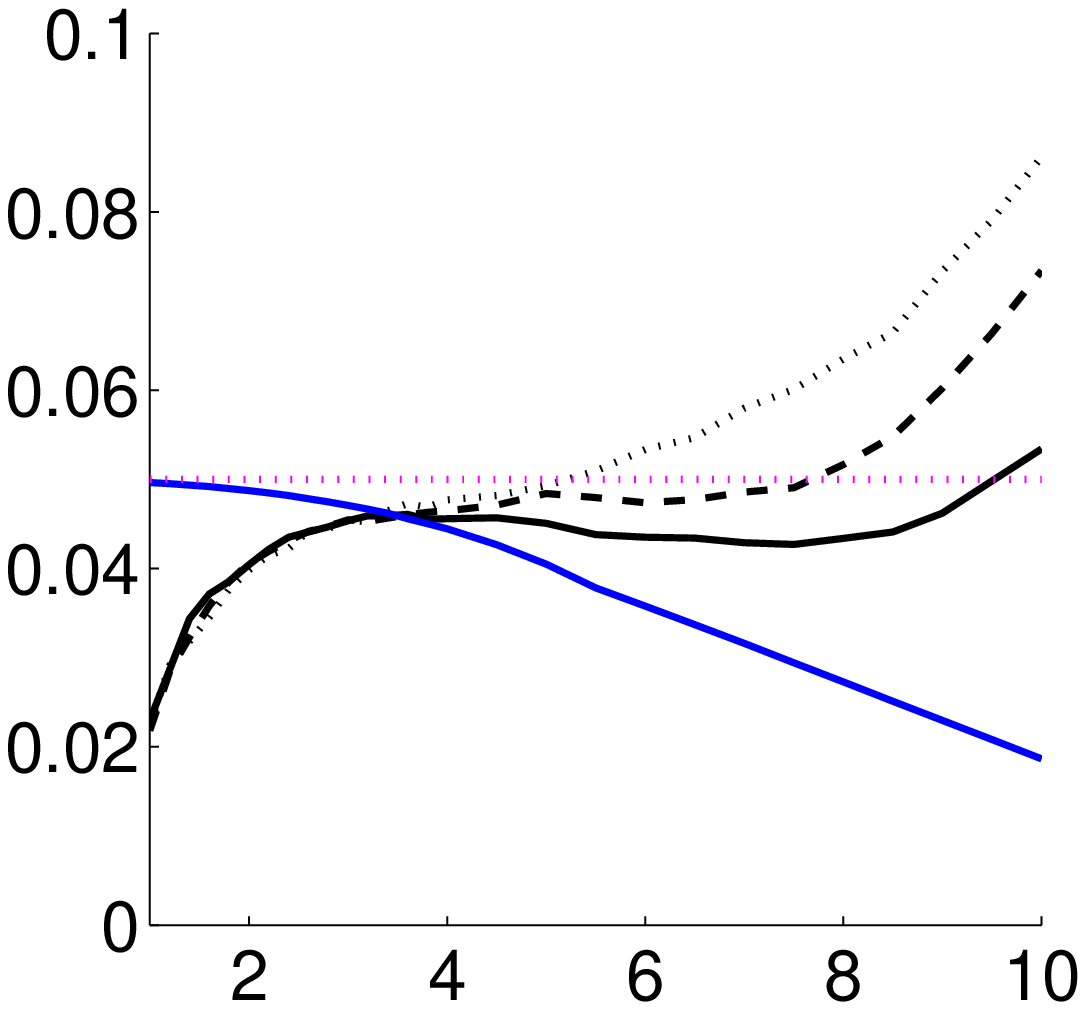} &
\includegraphics[trim=40 10 60 10,clip,width=1.4in]{\figurepath/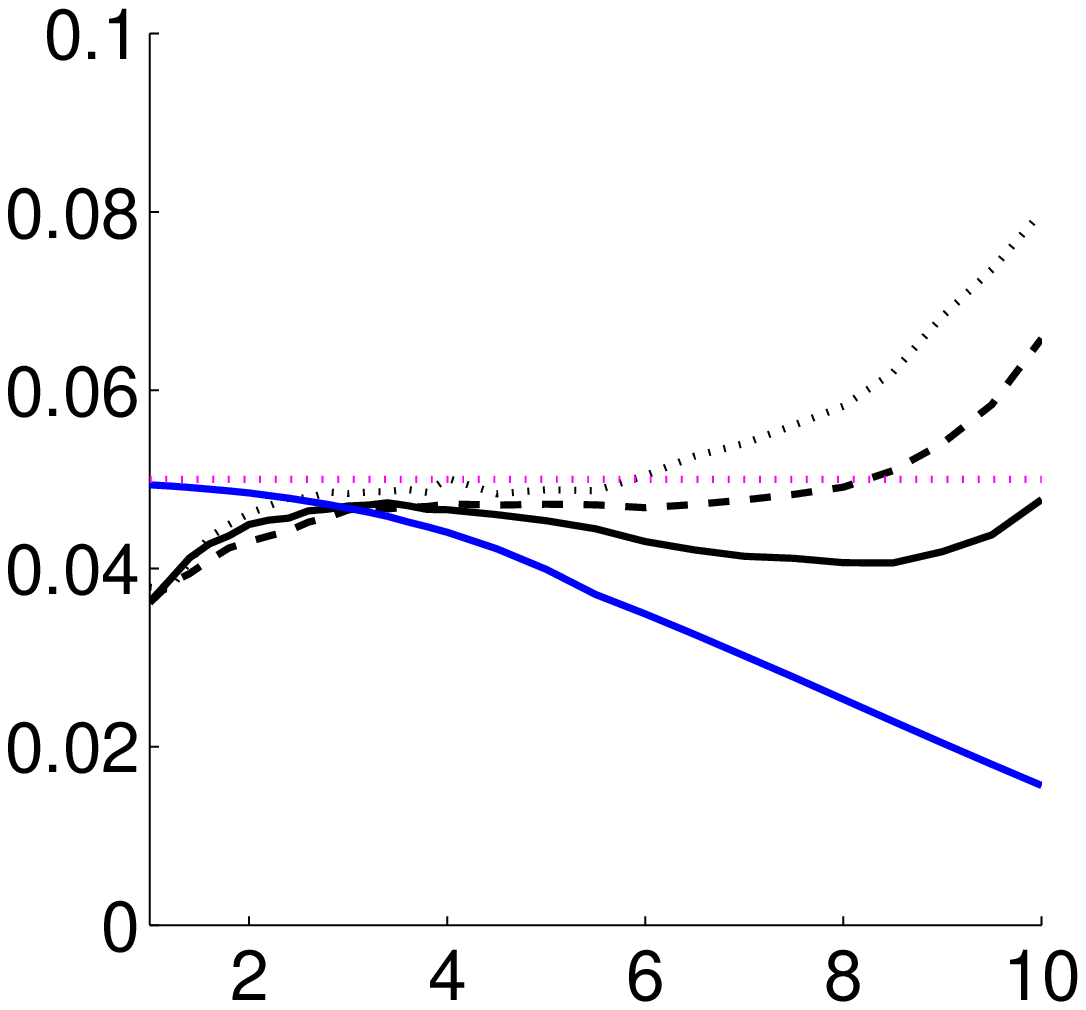} &
\includegraphics[trim=40 10 60 10,clip,width=1.4in]{\figurepath/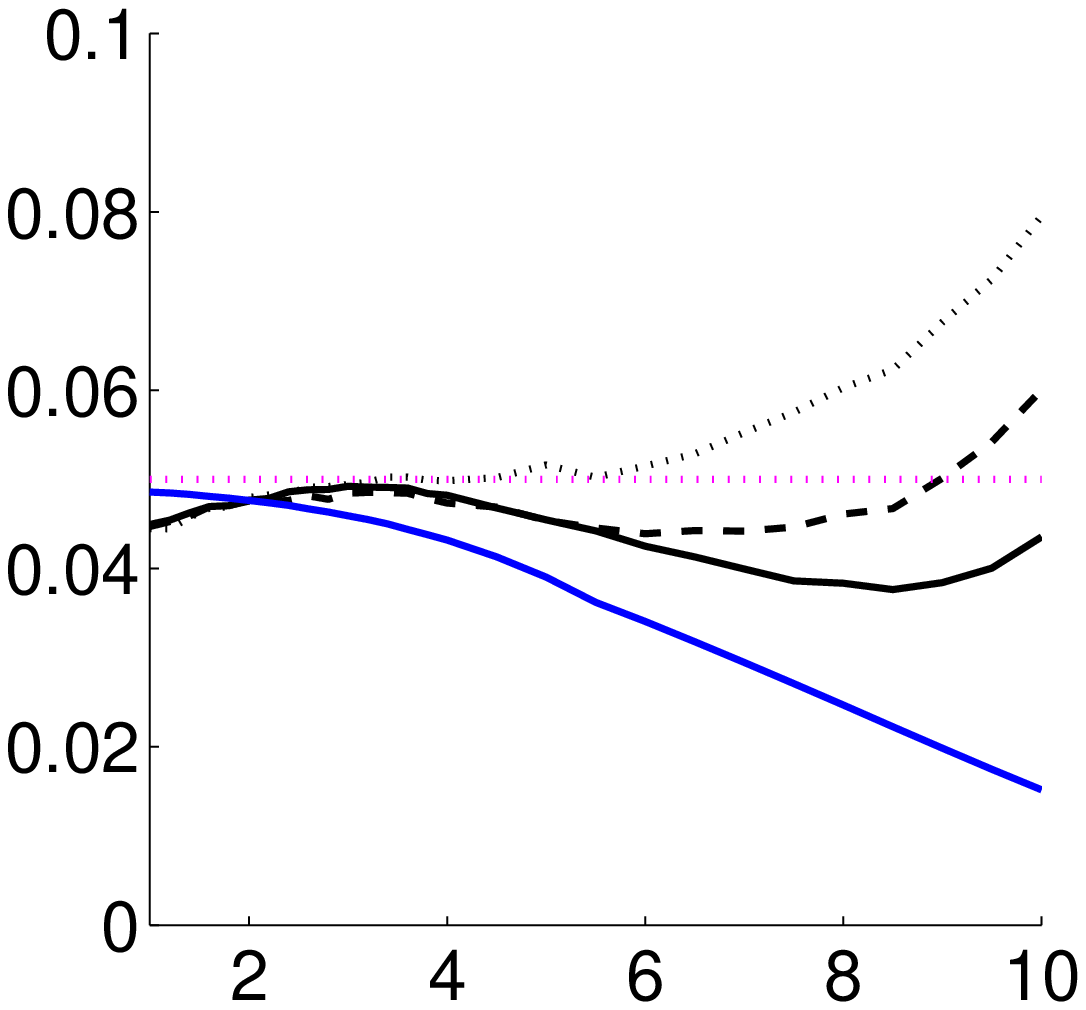} \\
\begin{sideways}\phantom{-----------}Power\end{sideways} &
\includegraphics[trim=40 10 60 10,clip,width=1.4in]{\figurepath/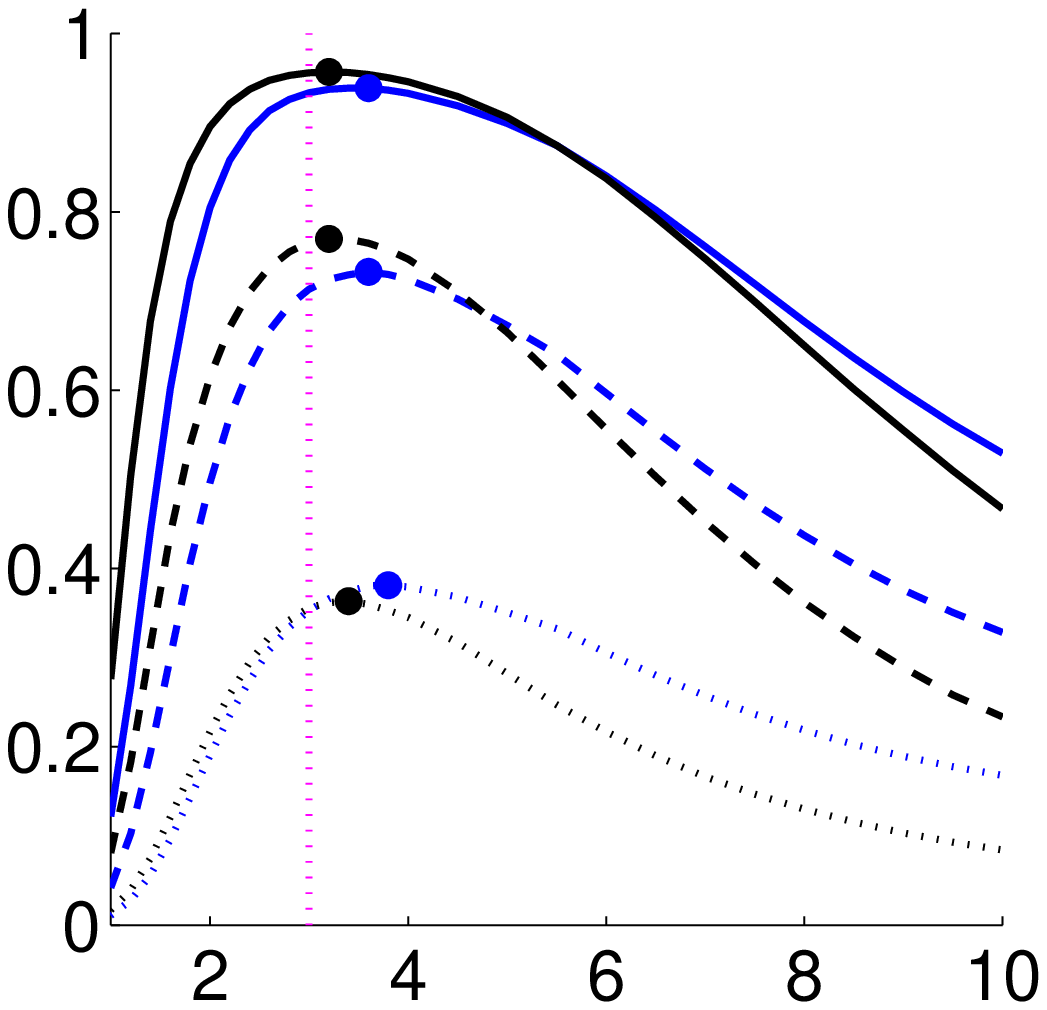} &
\includegraphics[trim=40 10 60 10,clip,width=1.4in]{\figurepath/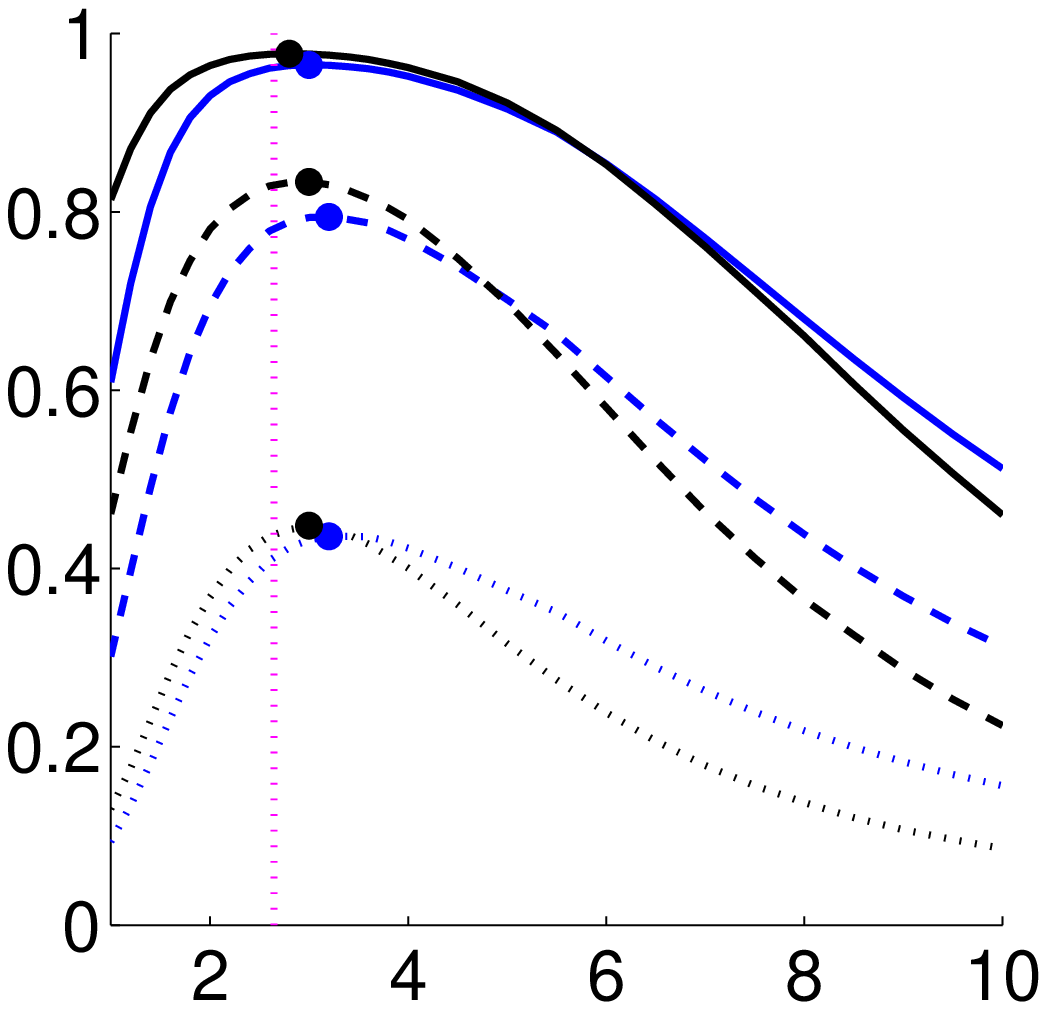} &
\includegraphics[trim=40 10 60 10,clip,width=1.4in]{\figurepath/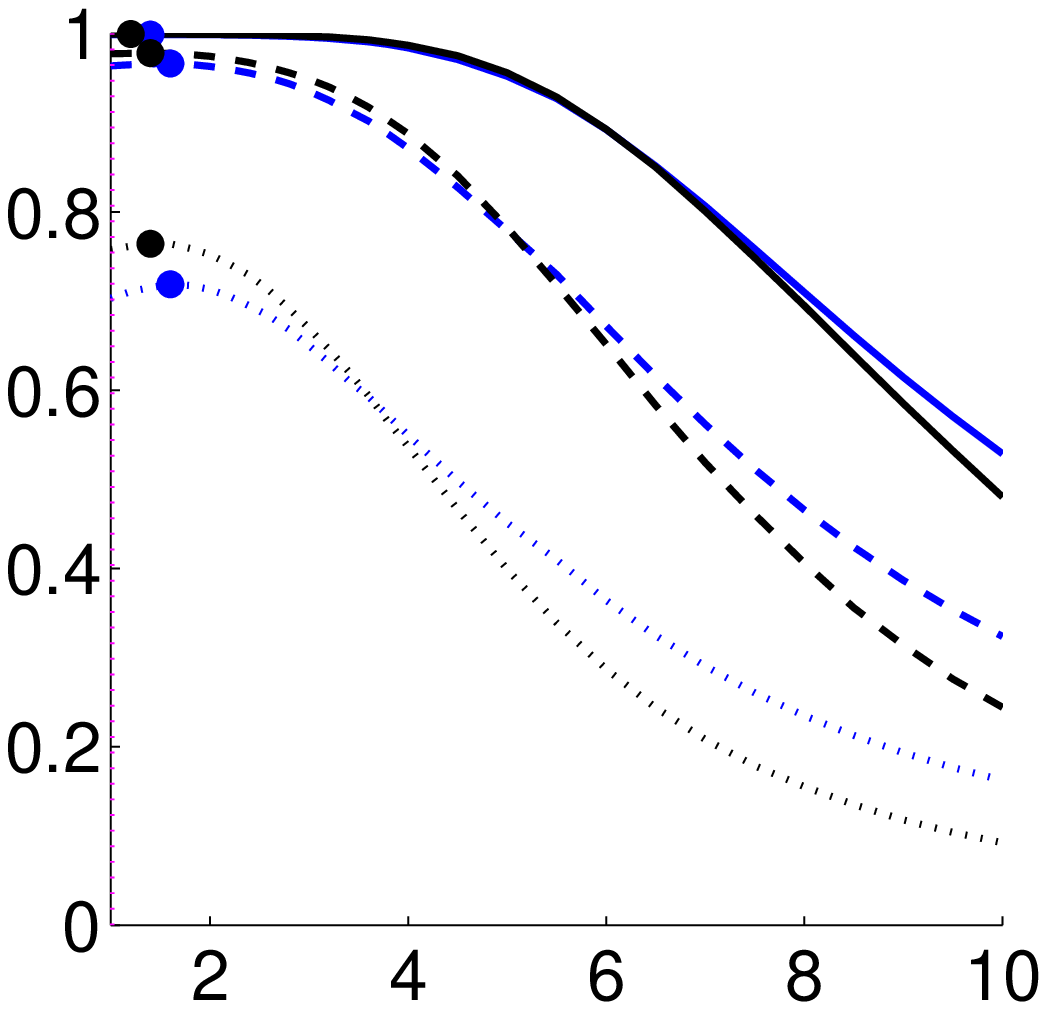} \\
& \phantom{---}$\gamma$ & \phantom{---}$\gamma$ & \phantom{---}$\gamma$
\end{tabular}
\caption{ \label{fig:error} Realized (black) and ``theoretical'' (blue) FDR, Realized (black) and ``theoretical'' (blue) power of the BH procedure by the exact height distribution $F_\gamma$ (i.e., $v=-\infty$) for $a=55$ (solid), $a=45$ (dashed) and $a=35$ (dotted). The maxima of the curves (solid circles) approach the optimal bandwidth (vertical dashed).}
 \end{center}
 \end{figure}

Simulations were used to evaluate the performance and limitations of the STEM algorithm for finite range $L=200$, finite number of peaks $J=9$ and moderate signal strength $a$ over $\R^2$ (i.e., $N=2$). Adopting the notations in Example \ref{ex:Gaussian}, the truncated Gaussian peaks $a_jh_j(t)$ are constructed with $a_j=a$, $b_j=3$ and $c_j=3$ for all $j=1,\ldots,J$ and varying $a$, and $\{\tau_j\}_{1\le j\le J} = \{(50i_1, 50i_2)\}_{i_1,i_2=1,2,3}$; the noise $z(t)$ is constructed with $\sigma=1$ and varying $\nu$; the smoothing kernel $w_\gamma(t)$ is constructed with $c=3$ and varying $\gamma$. The noise parameters $\sigma_\gamma$, $\rho_\gamma'$ and $\rho_\gamma''$ (note that $\kappa_\gamma=1$) were estimated using the same smoothing kernel. The BH procedures were applied at level $\alpha=0.05$ and over 10,000 replications to simulate the expectations.

Figure \ref{fig:error} shows the realized FDR and power of the BH procedures by the STEM algorithm, evaluated according to \eqref{eq:true-FDR} and \eqref{eq:true-power} with $v=-\infty$. As predicted by the theory, for every fixed bandwidth $\gamma$, the FDR is controlled below $\alpha= 0.05$ for strong enough signal $a$, and the power increases to 1. The theoretical FDR curve (blue) is evaluated according to the upper bound in \eqref{eq:bound-FDR}, while the theoretical power curve (blue) is derived by plugging the asymptotic threshold $u^*_{\BH}(-\infty)$ \eqref{eq:thresh-BH-fixed} into the approximated power \eqref{eq:approx-power}. It can be seen that the realized FDR fits the theoretical when the smoothing bandwidth $\gamma$ is close to the optimal one. However, for large $\gamma$, the realized FDR increases because of signal contamination in the transition region $\mathbb{T}_{\gamma}$. This phenomenon goes away as the signal $a$ increases. We also find that for small $\nu$ and small $\gamma$, say $\nu=0$ and $\gamma=1$, the difference between the realized FDR and theoretical FDR is relatively large. This is because the smoothed field $z_\gamma$ is not smooth enough in such case. Similar phenomena appear for the power. The simulation shows that when the signal gets stronger, the bandwidth maximizing the realized power gets closer to the optimal $\gamma$.

\begin{figure}[t]
\begin{center}
\begin{tabular}{cccc}
& $\nu=0$ & $\nu=1$ & $\nu=2$ \\
\begin{sideways}\phantom{-----------}FDR\end{sideways} &
\includegraphics[trim=40 10 60 10,clip,width=1.4in]{\figurepath/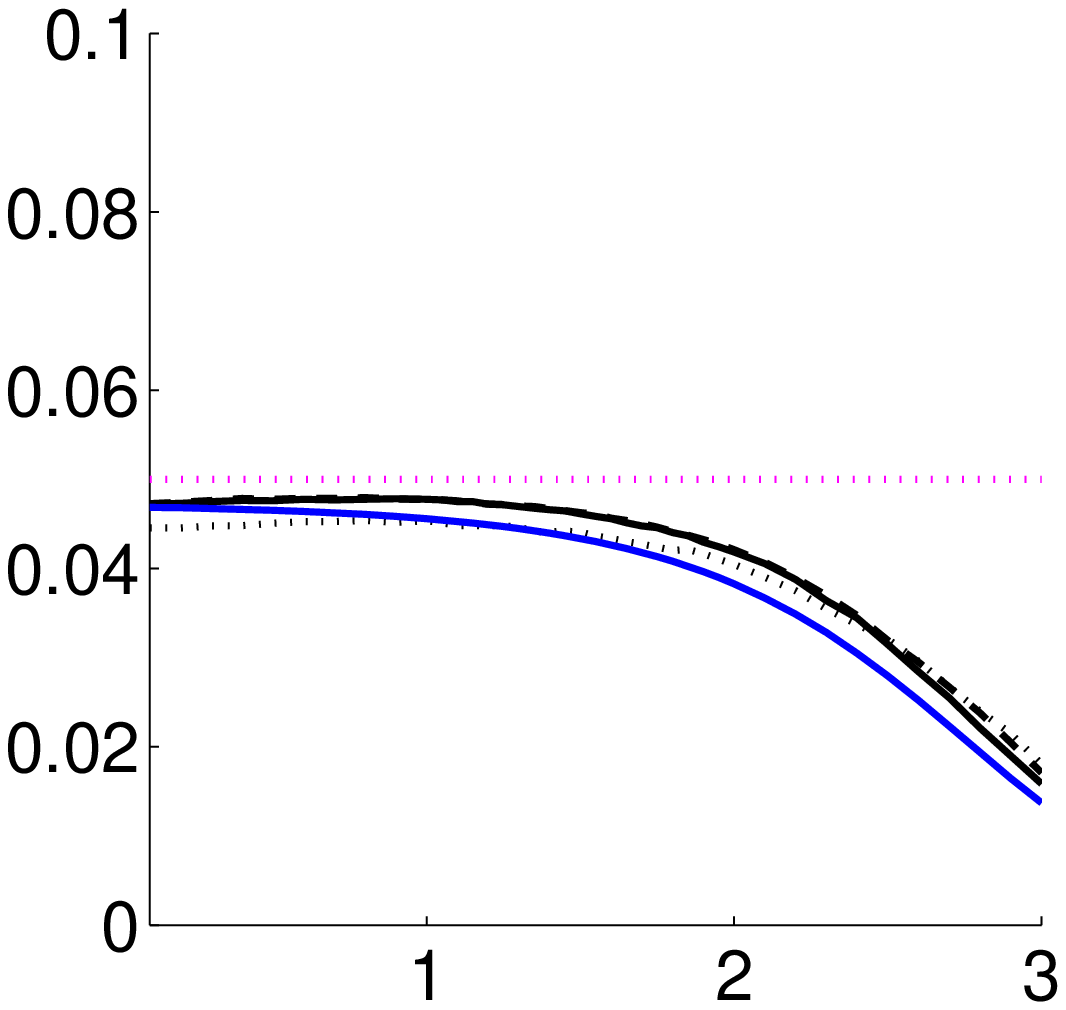} &
\includegraphics[trim=40 10 60 10,clip,width=1.4in]{\figurepath/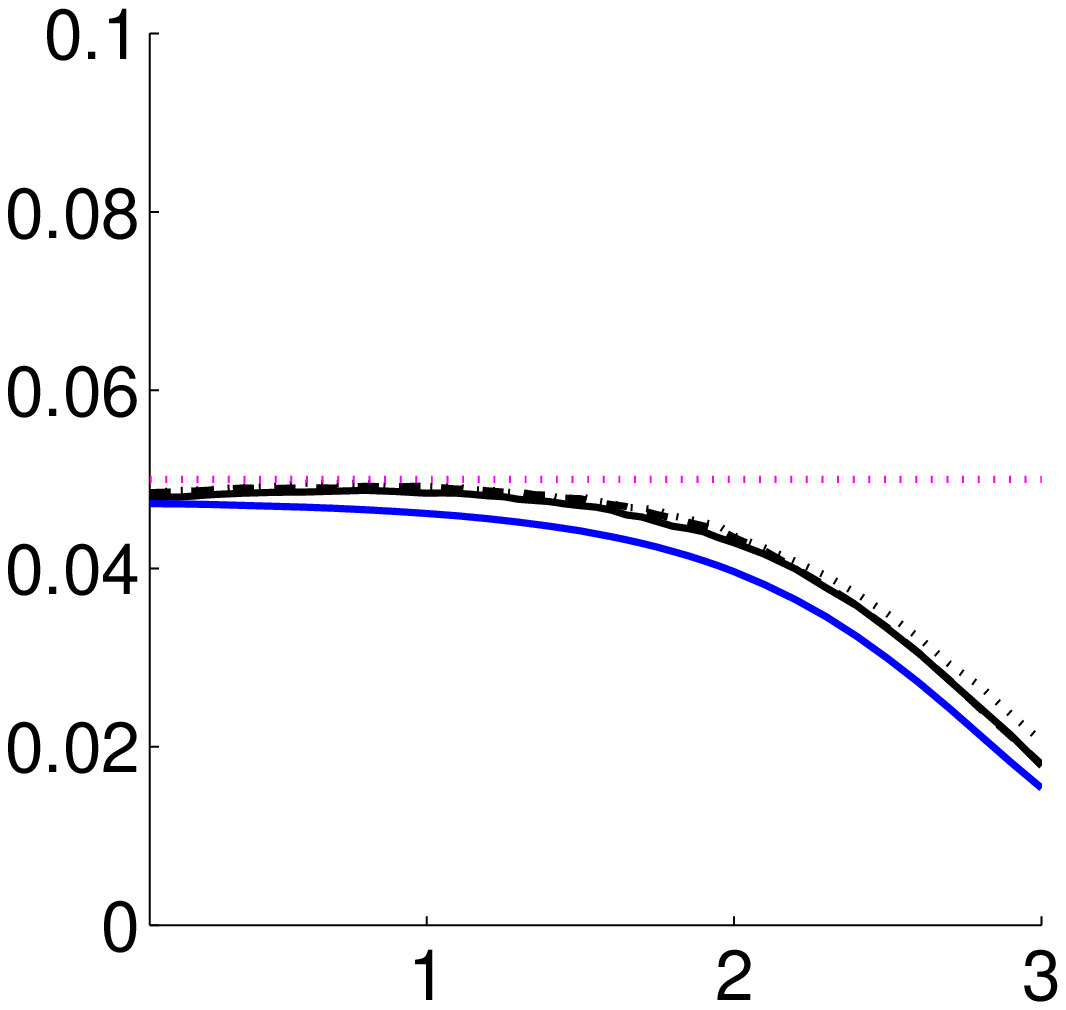} &
\includegraphics[trim=40 10 60 10,clip,width=1.4in]{\figurepath/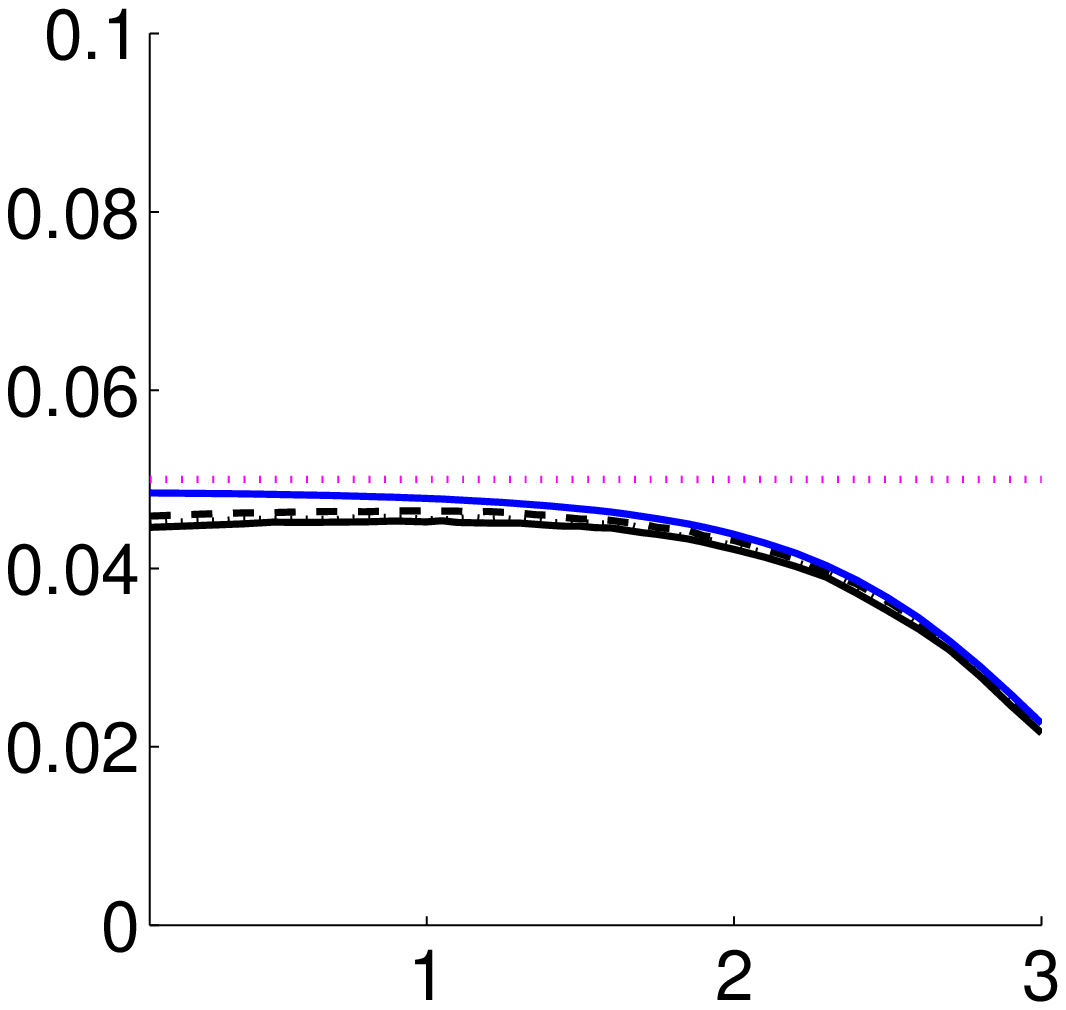} \\
\begin{sideways}\phantom{-----------}Power\end{sideways} &
\includegraphics[trim=40 10 60 10,clip,width=1.4in]{\figurepath/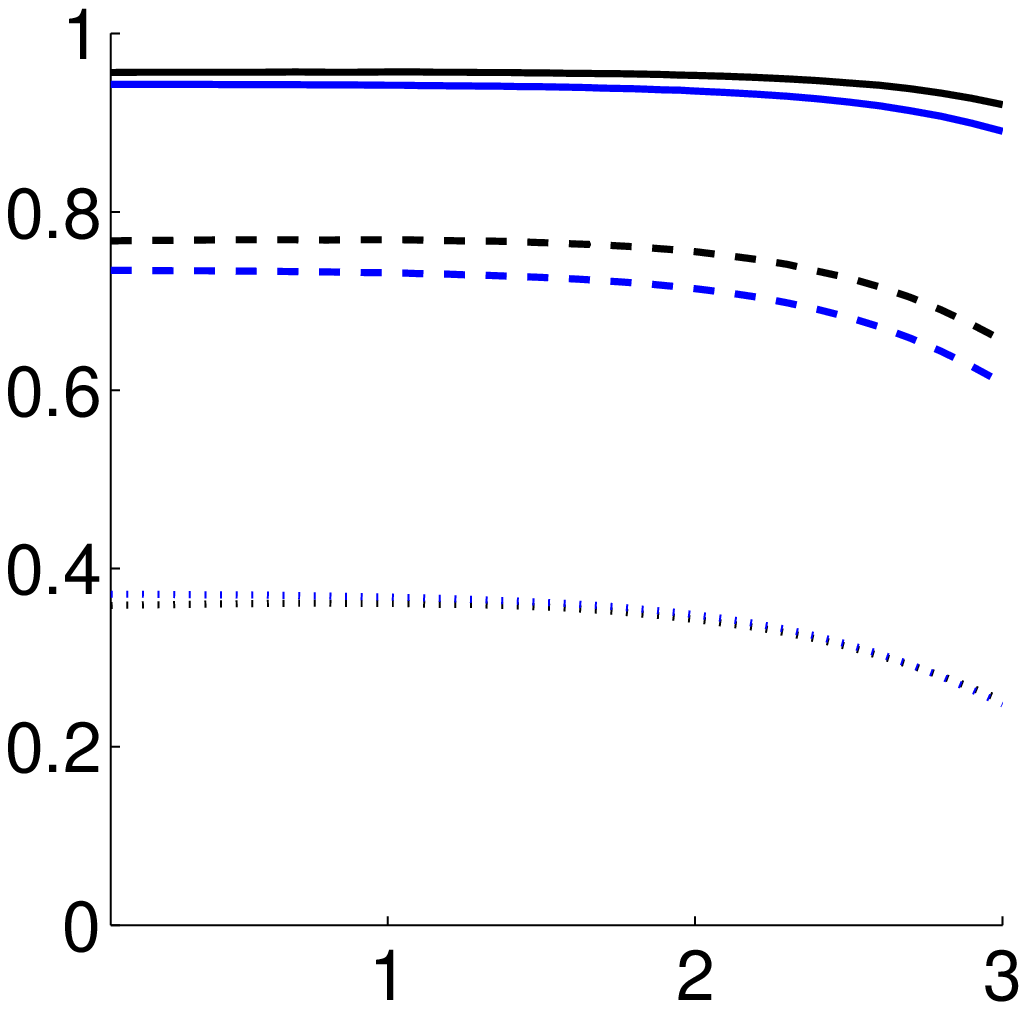} &
\includegraphics[trim=40 10 60 10,clip,width=1.4in]{\figurepath/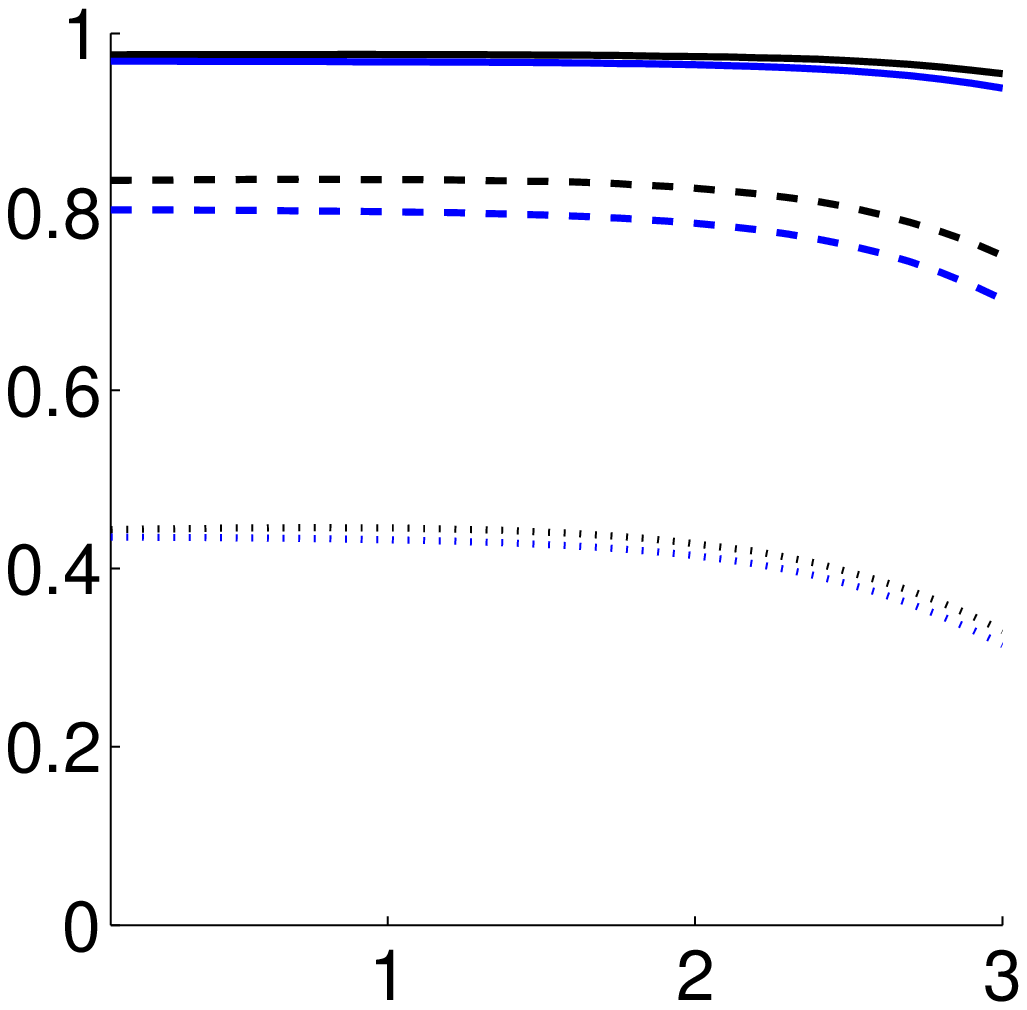} &
\includegraphics[trim=40 10 60 10,clip,width=1.4in]{\figurepath/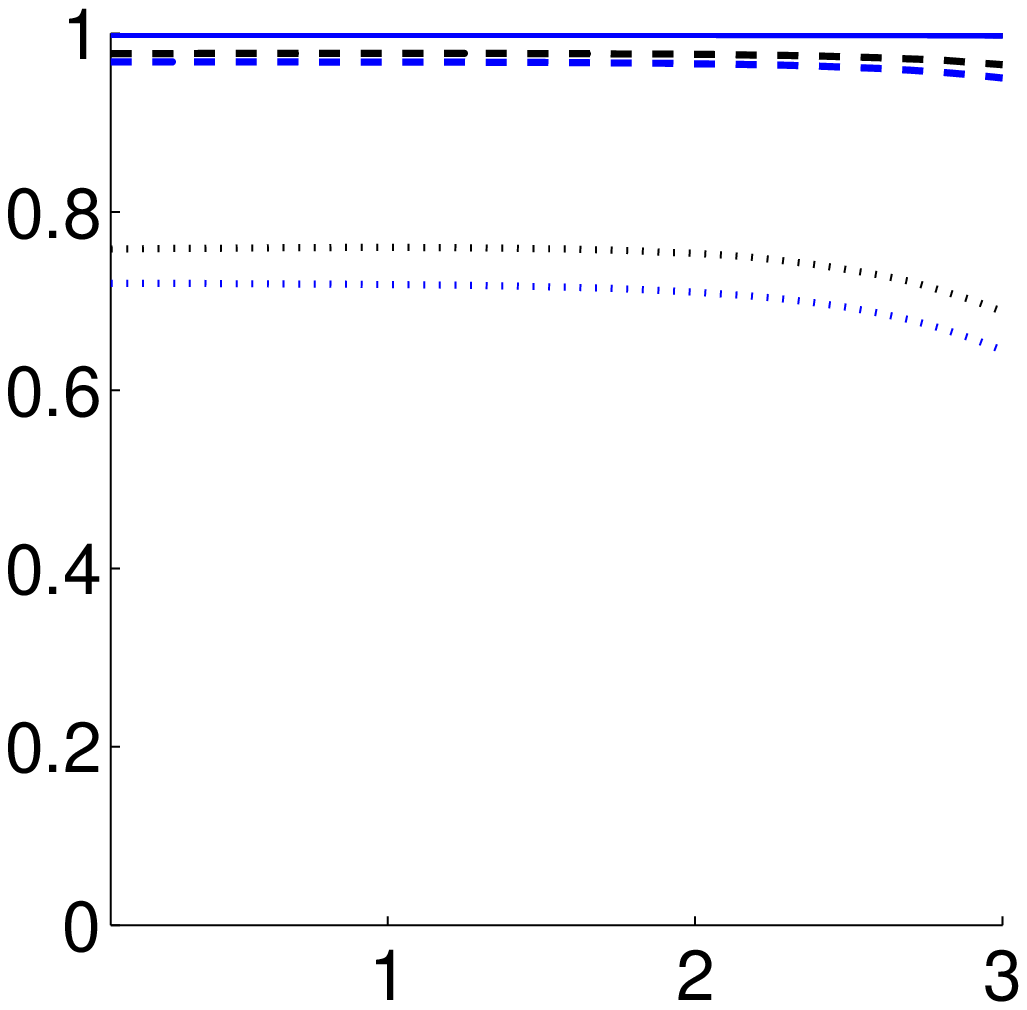} \\
& \phantom{---}$v/\sigma_\gamma$ & \phantom{---}$v/\sigma_\gamma$ & \phantom{---}$v/\sigma_\gamma$
\end{tabular}
\caption{ \label{fig:error-os} Realized (black) and ``theoretical'' (blue) FDR, Realized (black) and ``theoretical'' (blue) power of the BH procedure by the exact overshoot distribution $F_\gamma(\cdot,v)$ for $a=55$ (solid), $a=45$ (dashed) and $a=35$ (dotted).}
 \end{center}
 \end{figure}

Figure \ref{fig:error-os} shows the realized FDR and power of the BH procedures by the STEM algorithm, using the exact overshoot distribution $F_\gamma(\cdot,v)$ to compute p-values. Here, the bandwidth is chosen to be the optimal $\gamma$ and is fixed. The theoretical FDR curve (blue) is evaluated according to the upper bound in \eqref{eq:bound-FDR}, while the theoretical power curve (blue) is derived by plugging the asymptotic threshold $u^*_{\BH}(v)$ \eqref{eq:thresh-BH-fixed} into the approximated power \eqref{eq:approx-power}. As shown in Figure \ref{fig:error-os}, as the pre-threshold $v$ gets larger, the FDR becomes smaller and so does the power. This confirms the observation made after \eqref{eq:true-power} that the case of $\nu=-\infty$ gives the best performance if the exact height distribution $F_\gamma$ is known. However, when the signal is relatively strong, pre-thresholding does not weaken the power too much.

\begin{figure}[t]
\begin{center}
\begin{tabular}{cccc}
& $\nu=0$ & $\nu=1$ & $\nu=2$ \\
\begin{sideways}\phantom{-----------}FDR\end{sideways} &
\includegraphics[trim=40 10 60 10,clip,width=1.4in]{\figurepath/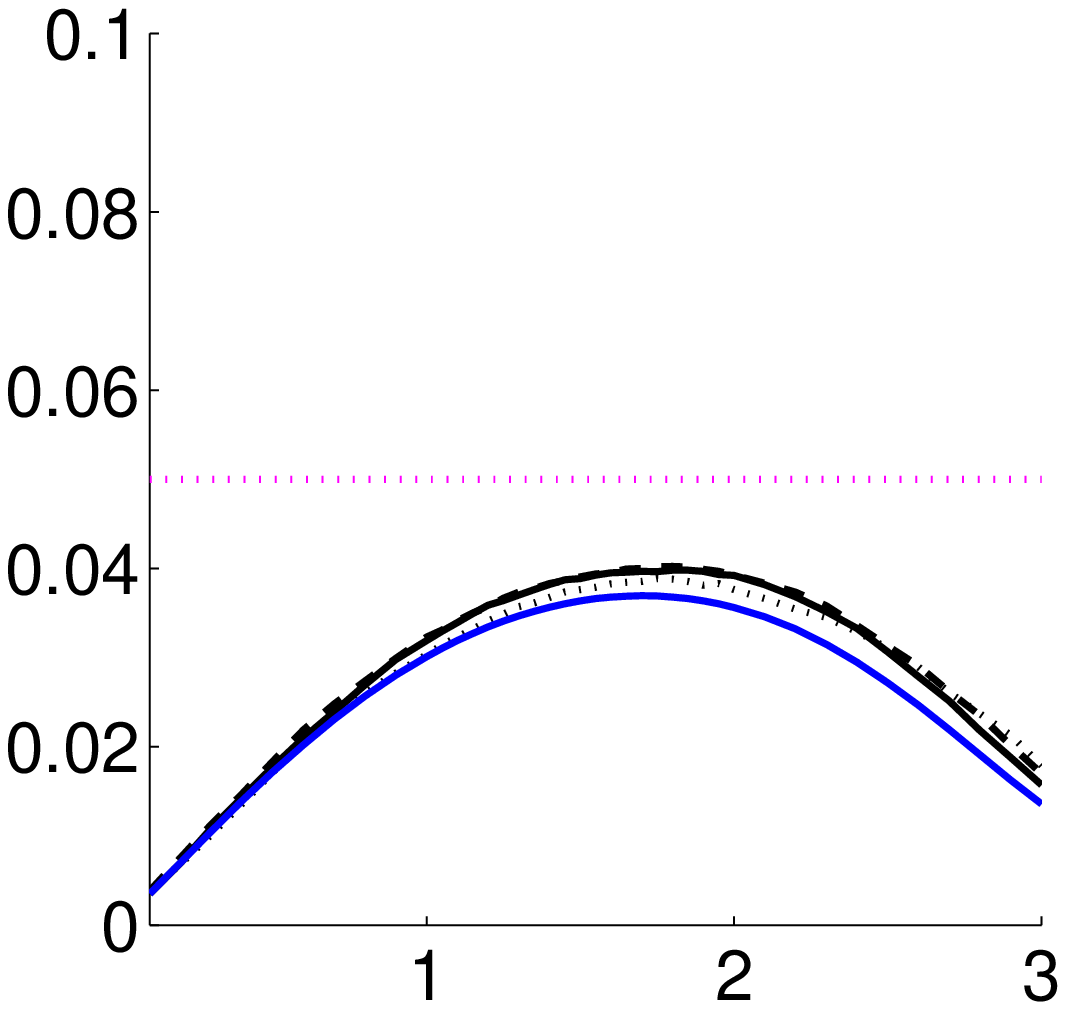} &
\includegraphics[trim=40 10 60 10,clip,width=1.4in]{\figurepath/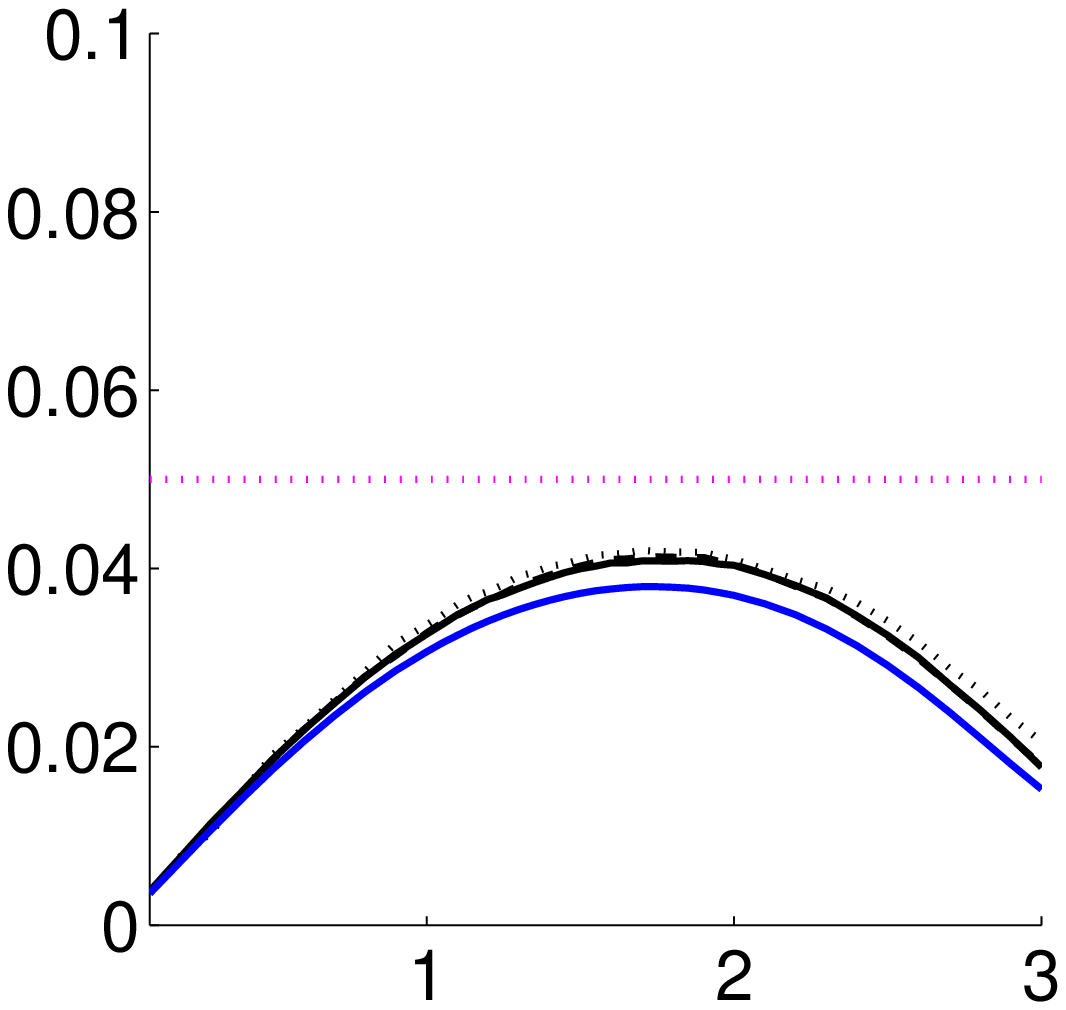} &
\includegraphics[trim=40 10 60 10,clip,width=1.4in]{\figurepath/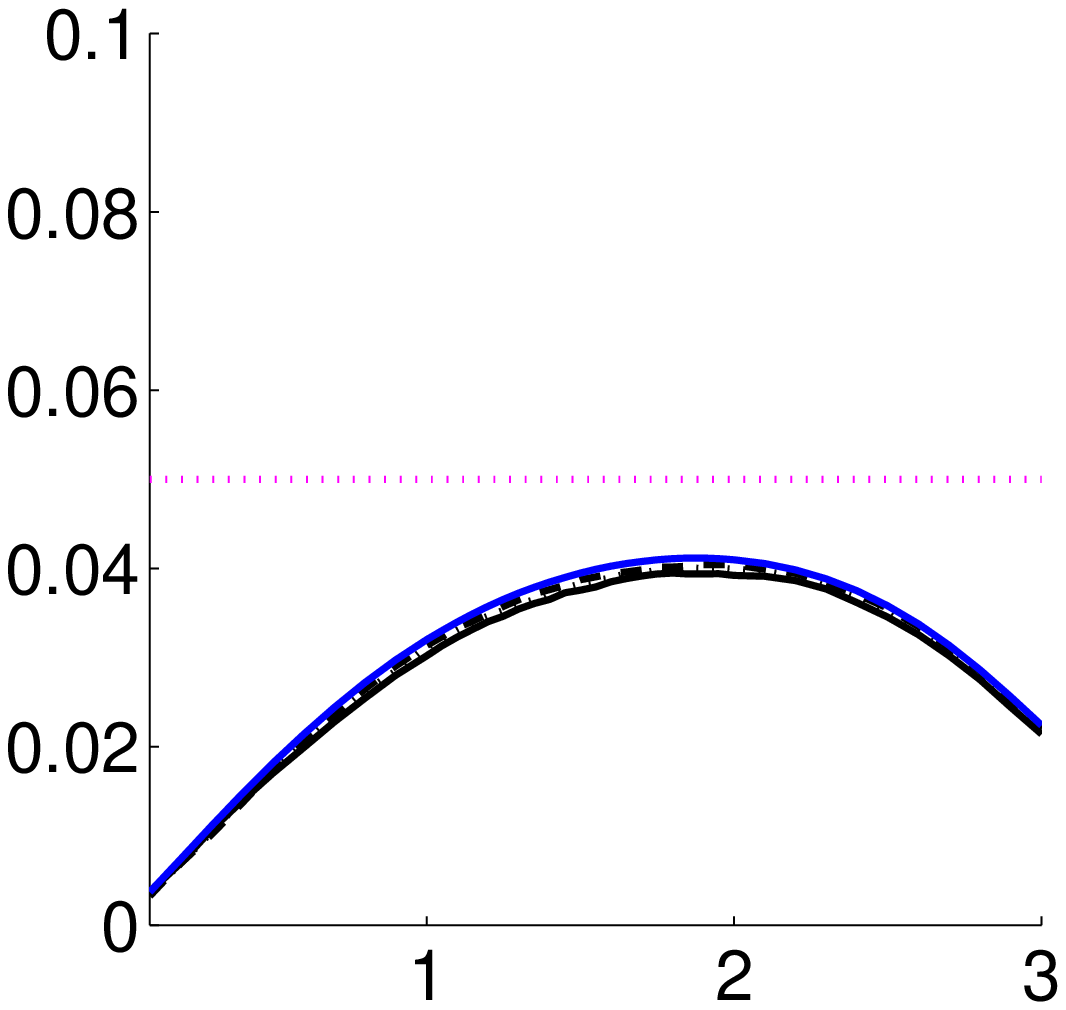} \\
\begin{sideways}\phantom{-----------}Power\end{sideways} &
\includegraphics[trim=40 10 60 10,clip,width=1.4in]{\figurepath/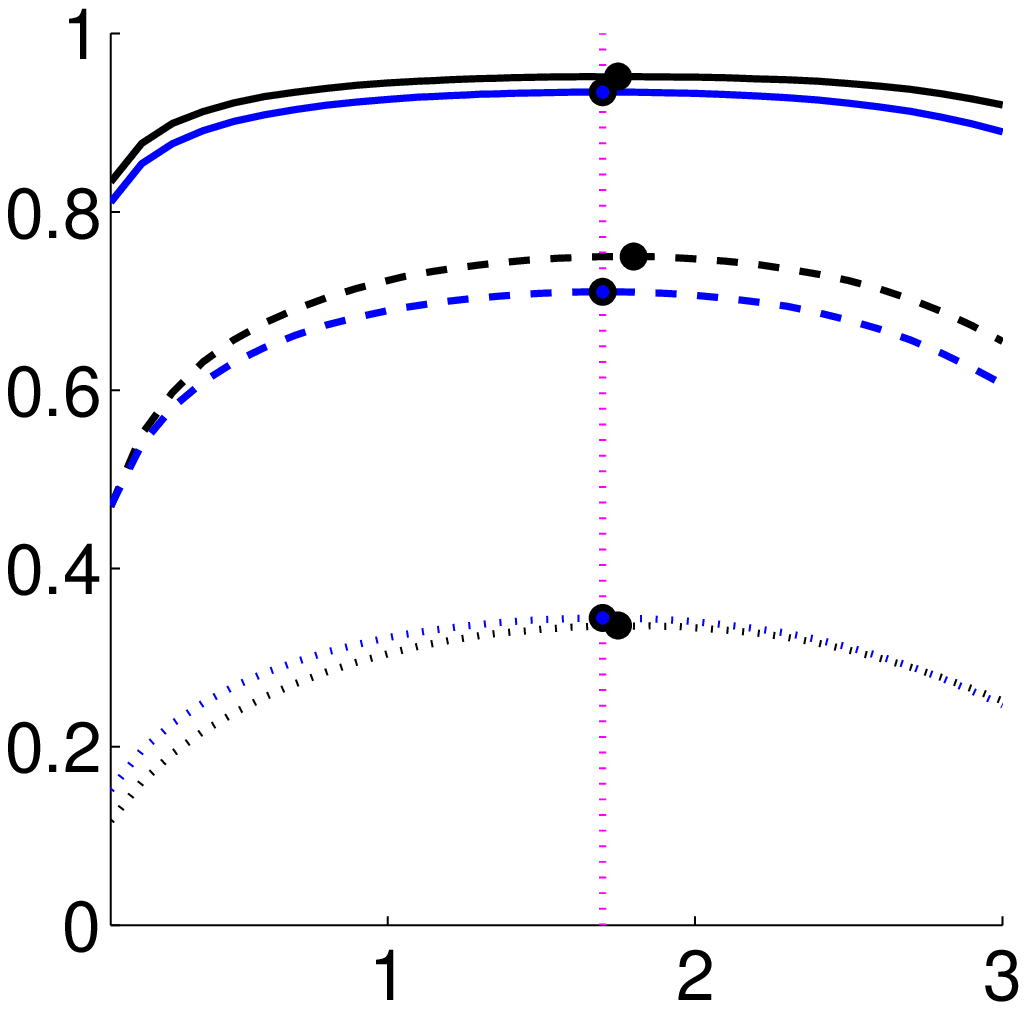} &
\includegraphics[trim=40 10 60 10,clip,width=1.4in]{\figurepath/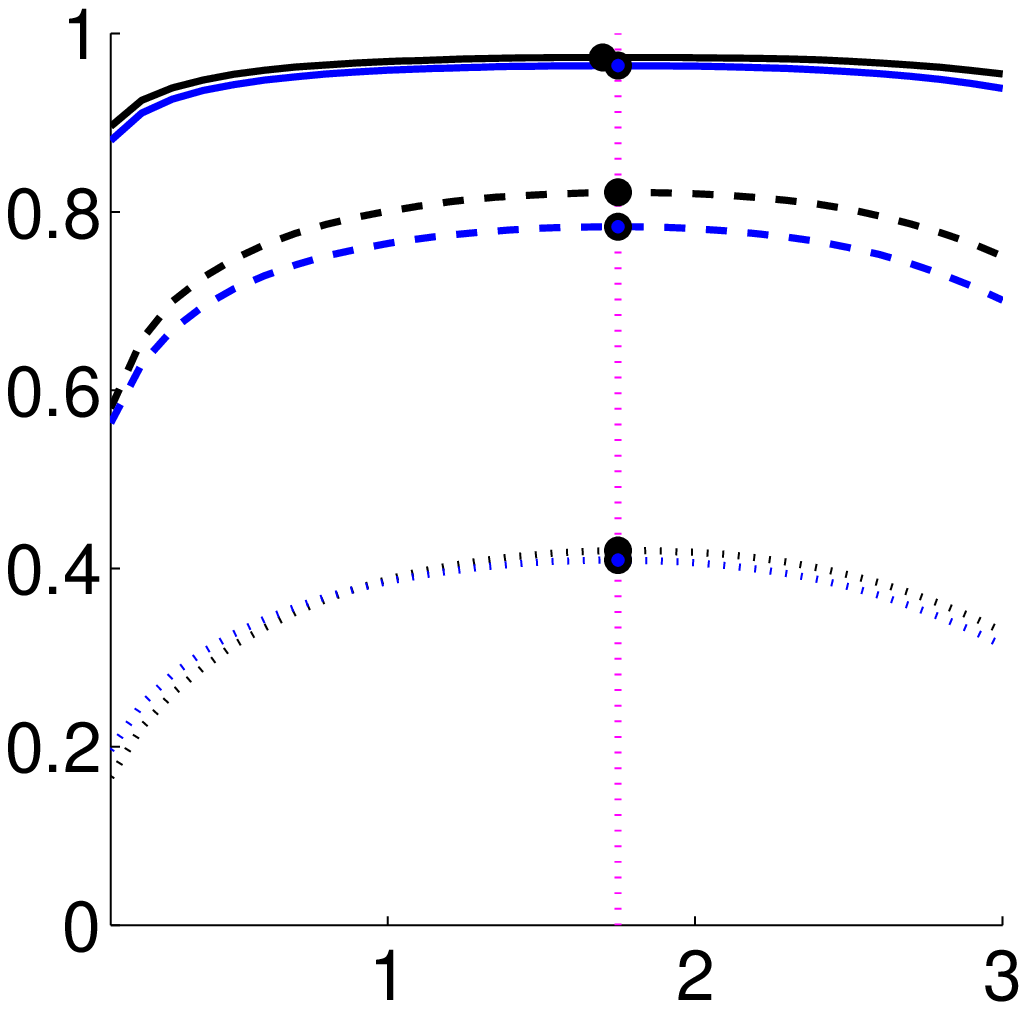} &
\includegraphics[trim=40 10 60 10,clip,width=1.4in]{\figurepath/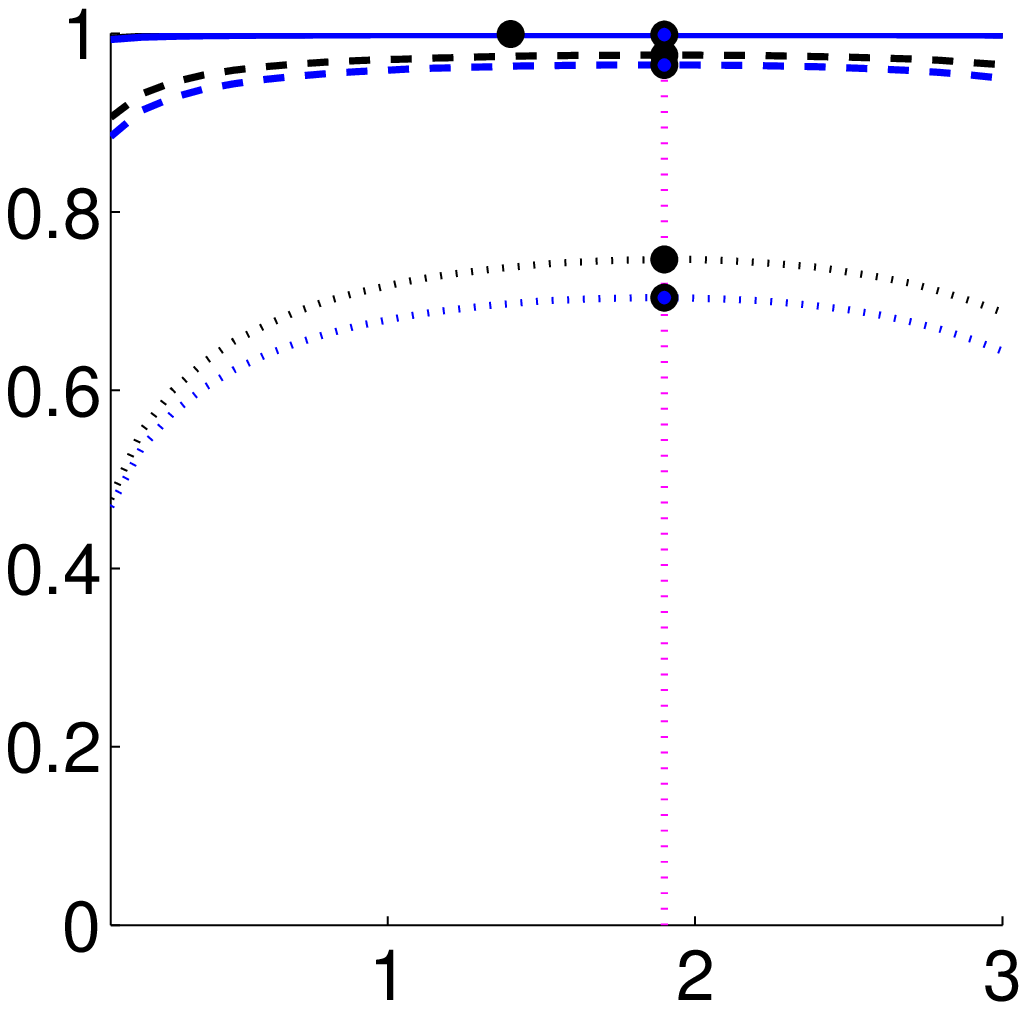} \\
& \phantom{---}$v/\sigma_\gamma$ & \phantom{---}$v/\sigma_\gamma$ & \phantom{---}$v/\sigma_\gamma$
\end{tabular}
\caption{ \label{fig:error-os-approx} Realized (black) and ``theoretical'' (blue) FDR, Realized (black) and ``theoretical'' (blue) power of the BH procedure by the approximated overshoot distribution $K_\gamma(\cdot,v)$ for $a=55$ (solid), $a=45$ (dashed) and $a=35$ (dotted). The maxima of the curves (solid circles) approach the optimal pre-threshold $v_{\rm opt}$ (vertical dashed).}
 \end{center}
 \end{figure}

Figure \ref{fig:error-os-approx} shows the realized FDR and power of the BH procedures by the STEM algorithm, using the approximated overshoot distribution $K_\gamma(\cdot,v)$ to compute p-values. Here again, the bandwidth is chosen to be the optimal $\gamma$ and is fixed. The theoretical FDR curve (blue) is evaluated according to the upper bound in \eqref{eq:bound-FDR-os-approx}, while the theoretical power curve (blue) is derived by plugging the asymptotic threshold $u^{**}_{\BH}(v)$ \eqref{eq:thresh-BH-fixed-os-approx} into the approximated power \eqref{eq:approx-power}. The simulation shows that the pre-threshold maximizing the realized power, which does not depend on the strength of the signal, is very close to the optimal pre-threshold $v_{\rm opt}$ \eqref{Eq:vopt}. Moreover, the realized curves still fit the theoretical curves very well for small $v$. This is because the limit in Theorem \ref{thm:FDR-os-approx} is in fact taken when the BH threshold is large.


\section{Data example}

\begin{figure}[t]
\begin{center}
\begin{tabular}{ccc}
\includegraphics[trim=35 10 35 0,clip,width=1.5in]{\figurepath/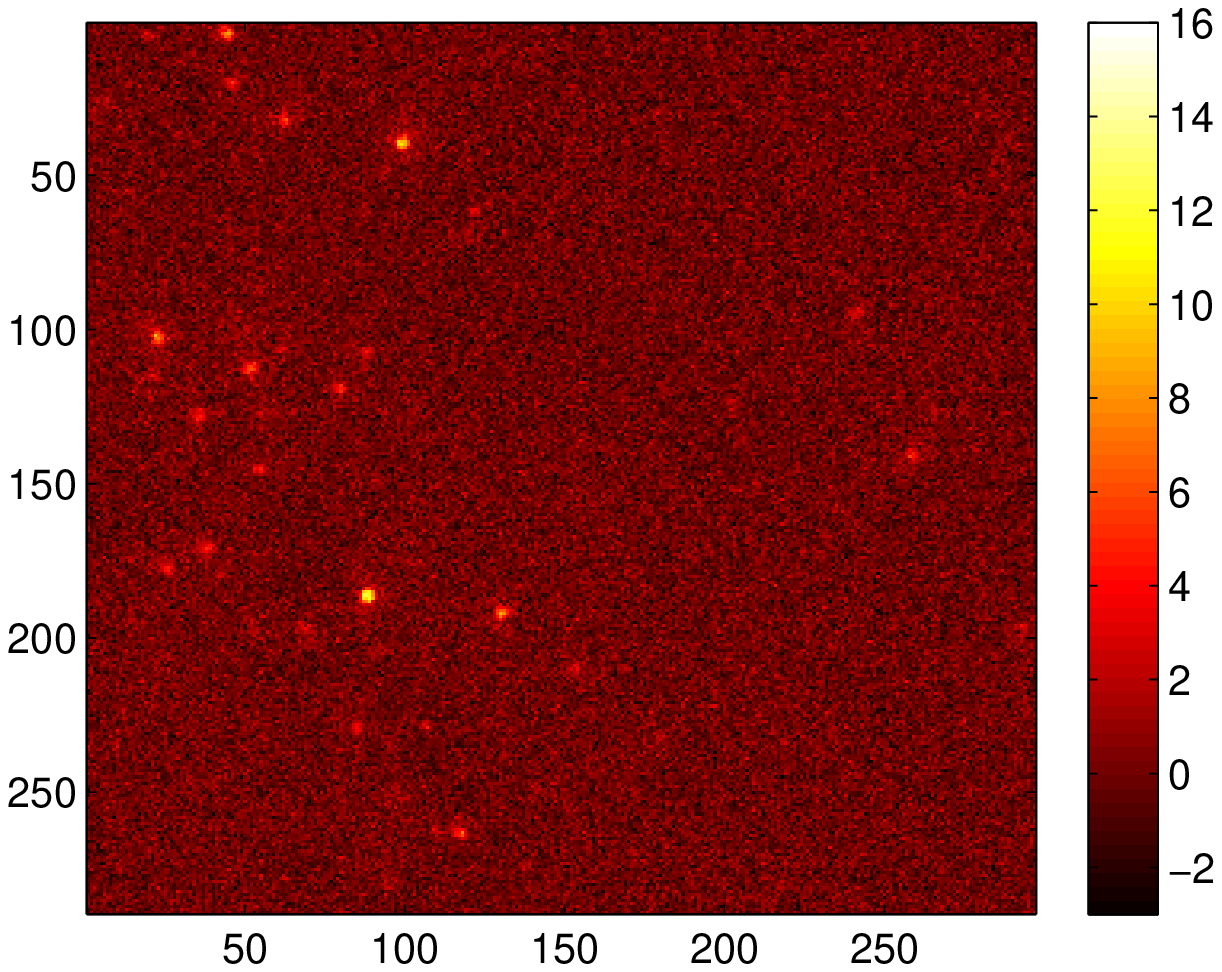} &
\includegraphics[trim=35 10 35 0,clip,width=1.5in]{\figurepath/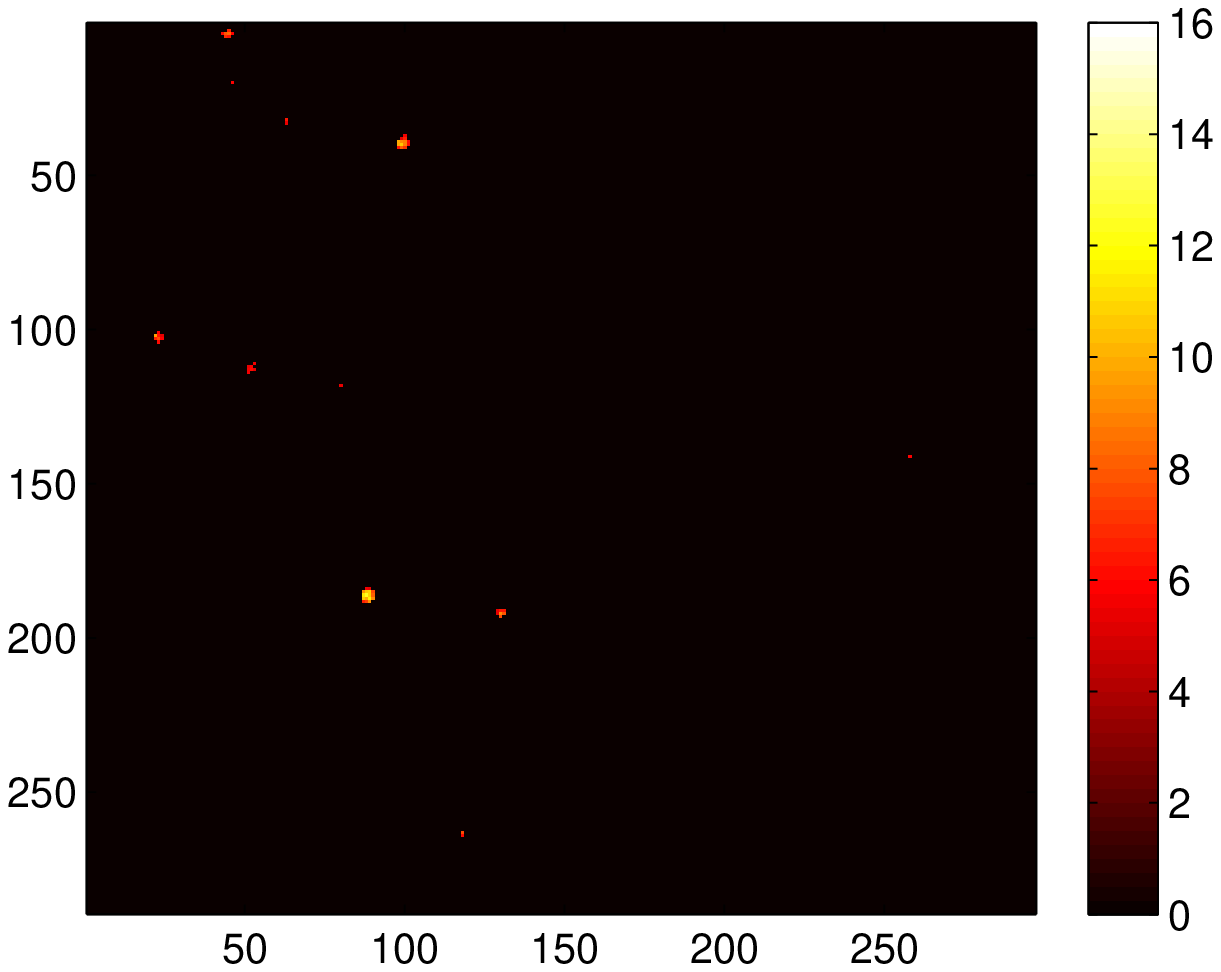}&
\includegraphics[trim=35 10 35 0,clip,width=1.5in]{\figurepath/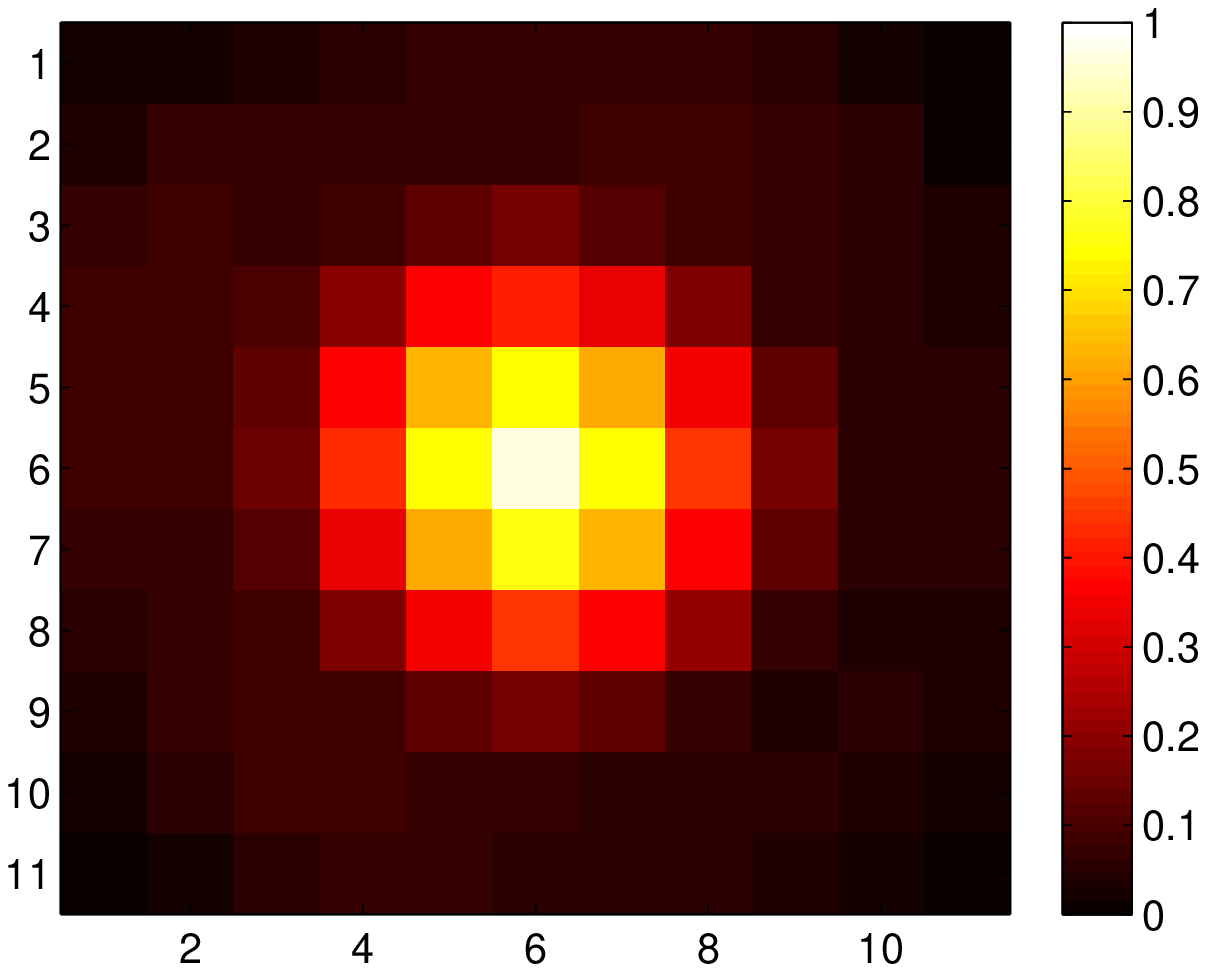}\\
\phantom{---}observed image & \phantom{---}pixelwise detection & \phantom{---}peak shape\\
\includegraphics[trim=35 10 35 0,clip,width=1.5in]{\figurepath/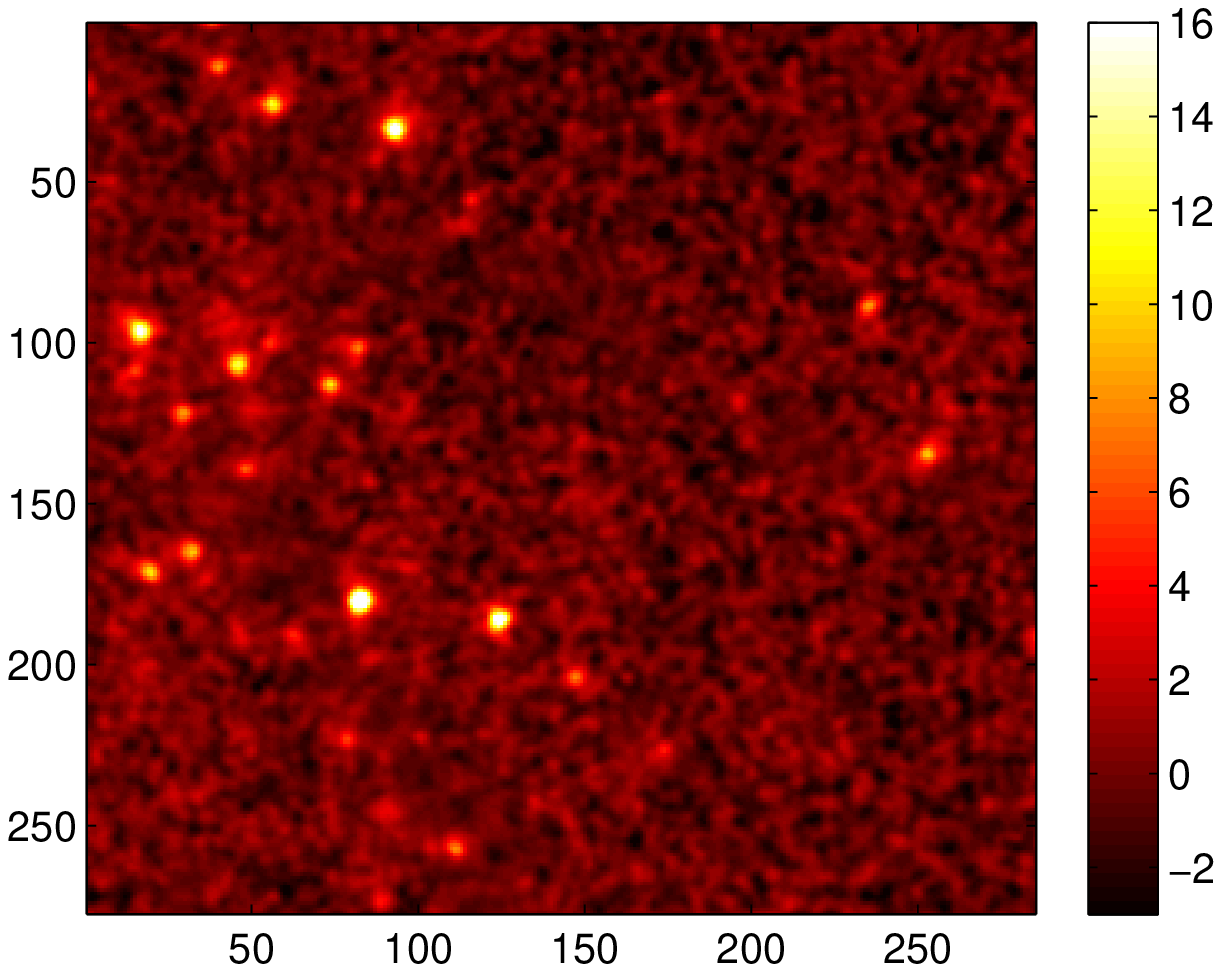}&
\includegraphics[trim=35 10 35 0,clip,width=1.5in]{\figurepath/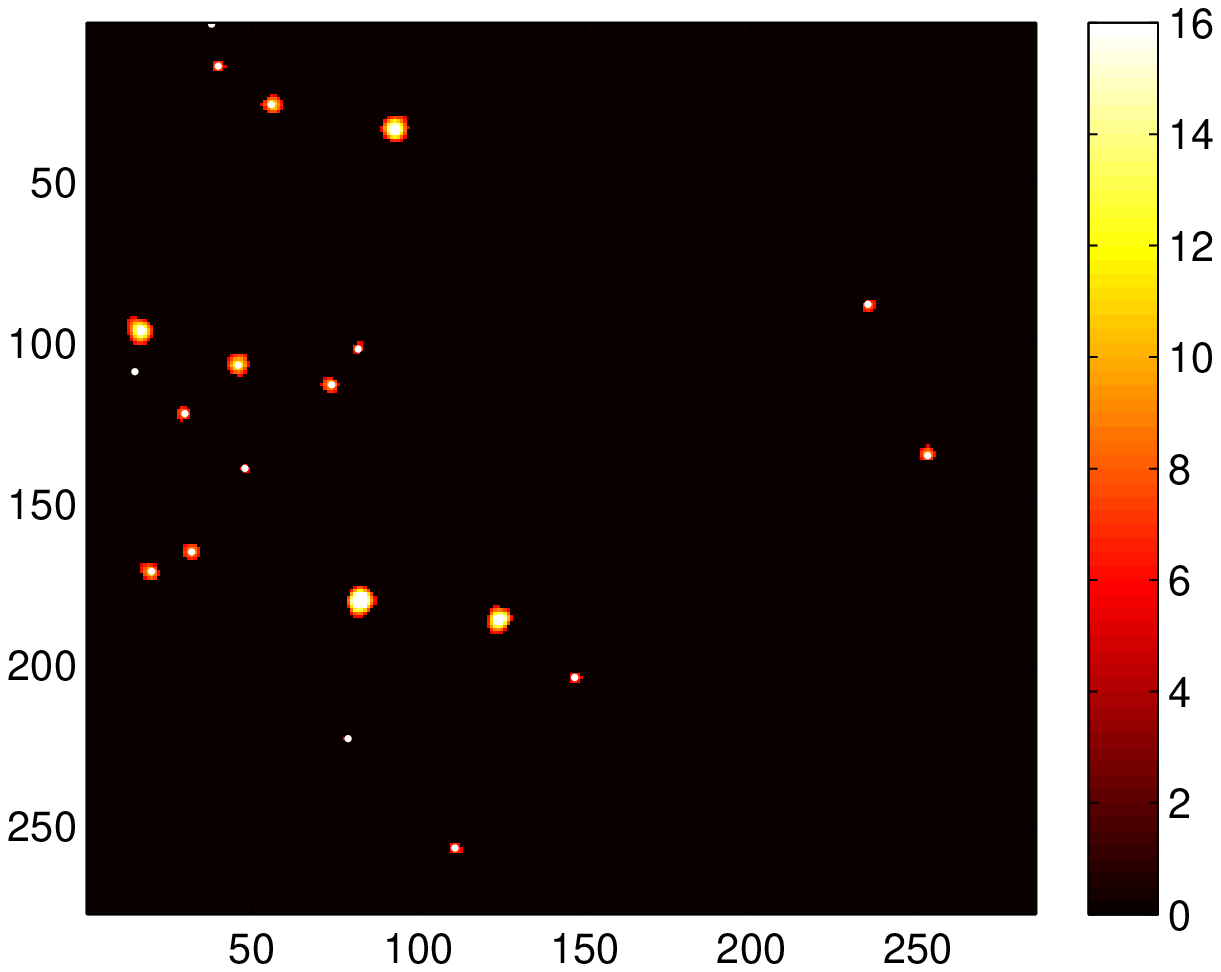} &
\includegraphics[trim=35 10 35 0,clip,width=1.5in]{\figurepath/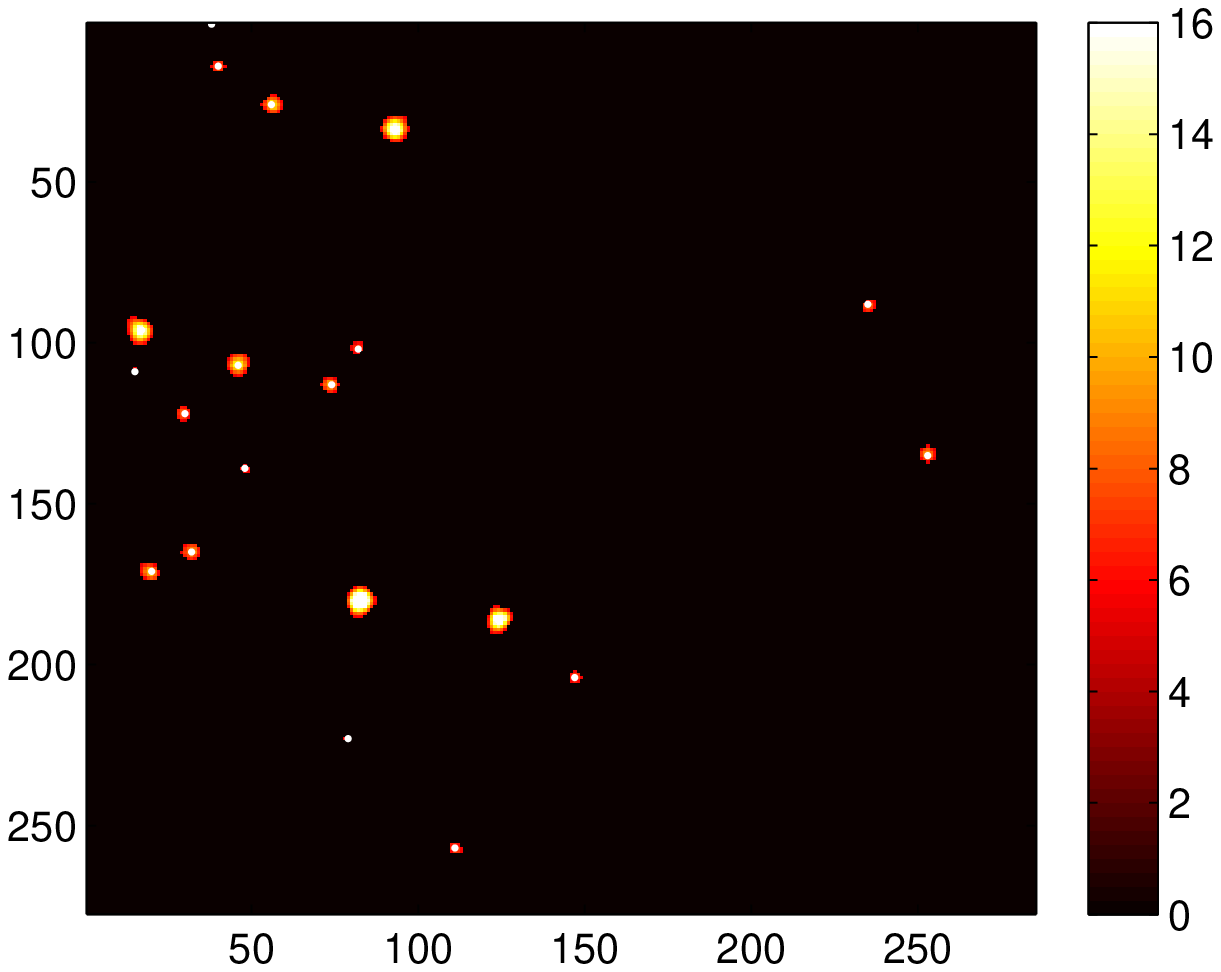}\\
 \phantom{---}smoothed image & \phantom{---}detection by height distr. & \phantom{---}detection by overshoot
\end{tabular}
\caption{ \label{fig:data} Data example.}
 \end{center}
 \end{figure}

The data consists of a stack of 1000 consecutive images of biological subcellular structure acquired via fluorescence nanoscopy \citep{Egner:2007,Geisler:2007}. The imaging technique, termed Photo-Activated Localization Microscopy with Independently Running Acquisition (PALMIRA), operates by shining an excitation laser at a very low intensity so that photons interact with only a small number of molecules at each recorded image frame. The recorded molecules appear as bright spots in each image. The image analysis task consists of separating those bright spots from the noisy background.

Figure \ref{fig:data} (top left) shows the first of those images, covering a region of about $10 \times 10~\mu m^2$. The background in this image has been log-transformed and adjusted by robust background estimates so that it can be assumed to have zero mean and unit variance. Pixelwise thresholding using the standard normal distribution for computing p-values and the BH algorithm on 85833 pixels per image times 1000 images at an FDR level of 0.0001 detects only the brightest regions (top middle) and provides a result where the FDR can only be interpreted in terms of pixels, not molecules.

The proper interpretation is given by the STEM algorithm. Considering the fluorescent molecules as point sources, the model of Section \ref{sec:model} captures the spatial extent of the signal peaks and the smoothness of the background field caused by dispersion of the recorded light in the acquisition process. The peak shape is seen in Figure \ref{fig:data} (top right), obtained as the average of the strongest peaks in the dataset, one from each image frame, aligned at their highest point. Robust estimation of the covariance function by rows and columns separately and across the image stack (not shown) indicates that the background noise may be modeled roughly by an isotropic Gaussian random field.

As shown in the simulation studies of Section \ref{sec:simulations}, a rough approximation to the peak shape suffices to obtain good detection power. To smooth the images (bottom left), we use an isotropic Gaussian smoothing kernel fitted by least squares to the estimated peak shape in the log domain, yielding $\gamma = 1.6$, and standardize the smoothed field so that the background has again mean zero and standard deviation $\sigma_\gamma = 1$. The choice of a Gaussian kernel allows using $\kappa_\gamma=1$ in the height distribution of Section \ref{sec:p-value-isotropic} for computing p-values, so that no correlation parameters need to be estimated. Thresholding the local maxima by the BH algorithm applied to the entire image stack at an FDR level of 0.0001 results in about 18 to 19 detected peaks in each frame (bottom middle).

Removing the isotropy assumption for the noise and computing p-values of the local maxima using the approximate overshoot distribution of Section \ref{sec:overshoot} instead with pre-threshold $v=1.8 \sigma_\gamma$ (close to optimal according to Figure \ref{fig:error-os-approx}) yields the same significant peaks (bottom right). This confirms the simulation results that using the approximate overshoot distribution yields a similar power, while being more general in its required assumptions.

\section{Technical details}
\label{sec:proofs}

\subsection{Supporting results}
\label{app:lemmas}


\begin{lemma}
\label{lemma:bounds}
Assume the model of Section \ref{sec:model} and let
$$
E=\sup_{\|x\|=1}{\rm Var} (\l \nabla z_{\gamma}(t), x \r), \quad F=\sup_{\|x\|=1}{\rm Var} (x^T \nabla^2 z_{\gamma}(t) x).
$$
For each $j$, let $u$ and $\ep_0>0$ be fixed, then for sufficiently large $a_j$,
\begin{equation}
\label{eq:I-bound}
\begin{split}
\P\left(\# \{t \in \tilde{T}\cap I_j^{\rm side}\} = 0\right)
&\ge
1 - \exp\left(-\frac{a_j^2C_j^2}{2E+\ep_0} \right), \\
\P\left(\# \{t \in \tilde{T}\cap I_j^{\rm mode}\} = 1\right)
&\ge
1 - \exp\left(-\frac{a_j^2C_j^2}{2E+\ep_0} \right) - \exp\left(-\frac{a_j^2D_j^2}{2F+\ep_0} \right), \\
\P\left(\# \{t \in \tilde{T}(u)\cap I_j^{\rm mode}\} = 1\right) &\ge
1 - \Phi\left(\frac{u-a_j M_j}{\sigma_\gamma}\right) - \exp\left(-\frac{a_j^2D_j^2}{2F+\ep_0} \right).
\end{split}
\end{equation}
\end{lemma}

\begin{proof}
(1) Consider first the side region $I_j^{\rm side}$. The probability that there are no local maxima of $y_\gamma(t)$ in $I_j^{\rm side}$ is greater than the probability that $\l \nabla y_\gamma(t), \tau_j-t\r > 0$ for all $t \in I_j^{\rm side}$. This probability is bounded below by
\begin{equation*}
\begin{split}
\P&\left(\inf_{I_j^{\rm side}} \l \nabla z_\gamma(t), \tau_j-t\r > -\inf_{I_j^{\rm side}} \l \nabla \mu_\gamma(t), \tau_j-t\r \right)\\
&=1 - \P\left(\sup_{I_j^{\rm side}} -\l \nabla z_\gamma(t), (\tau_j-t)/\|\tau_j-t\| \r \ge a_j C_j \right)\\
&\ge 1 - \P\left(\sup_{t\in I_j^{\rm side}}\sup_{\|x\|=1} \l \nabla z_\gamma(t), x \r \ge a_j C_j \right).
\end{split}
\end{equation*}
Then the first line of \eqref{eq:I-bound} follows from the Borell-TIS inequality \cite[Eq. (4.1.1)]{Adler:2007} and the stationarity of $z_\gamma$.

(2) The probability that $y_\gamma(t)$ has no local maxima in $I_j^{\rm mode}$ is less than the probability that there exists some $t\in \partial I_j^{\rm mode}$ such that $\l \nabla y_\gamma(t), \tau_j-t\r \le 0$, for all $t\in \partial I_j^{\rm mode}$ satisfying $\l \nabla y_\gamma(t), \tau_j-t\r > 0$ would imply the existence of at least one local maximum in $I_j$. This probability is bounded above by
\begin{equation*}
\begin{aligned}
\P&\left(\inf_{\partial I_j^{\rm mode}} \l \nabla z_\gamma(t), \tau_j-t\r \le -\inf_{\partial I_j^{\rm mode}} \l \nabla \mu_\gamma(t), \tau_j-t\r \right)\\
&= \P\left(\sup_{\partial I_j^{\rm mode}} - \bigg\langle \nabla z_\gamma(t), \frac{\tau_j-t}{\|\tau_j-t\|} \bigg\rangle \ge \inf_{\partial I_j^{\rm mode}} \bigg\langle \nabla \mu_\gamma(t), \frac{\tau_j-t}{\|\tau_j-t\|} \bigg\rangle \right)\\
&\le \P\left(\sup_{t\in I_j^{\rm side}}\sup_{\|x\|=1} \l \nabla z_\gamma(t), x \r \ge a_jC_j \right),
\end{aligned}
\end{equation*}
where the last inequality is due to fact that $\partial I_j^{\rm mode}$ is contained in the closure of $I_j^{\rm side}$. Then by the Borell-TIS inequality, for any fixed $\ep_0>0$ and sufficiently large $a_j$,
\begin{equation}
\label{eq:I-mode-bound-0}
\P\left(\#\{t \in \tilde{T}\cap I_j^{\rm mode})\} = 0\right) \le \exp\left(-\frac{a_j^2C_j^2}{2E+\ep_0} \right).
\end{equation}

On the other hand, the probability that $y_\gamma(t)$ has more than one local maxima in $I_j^{\rm mode}$ is less than the probability that the largest eigenvalue of $\nabla^2 y_\gamma(t)$ is nonnegative for some $t$ in $I_j^{\rm mode}$. This probability is
\begin{equation*}
\begin{aligned}
\P\left(\sup_{t \in I_j^{\rm mode}} \sup_{\|x\|=1} x^T \nabla^2 y_\gamma(t) x \ge 0\right) \le\P\left(\sup_{t \in I_j^{\rm mode}} \sup_{\|x\|=1} x^T \nabla^2 z_\gamma(t) x \ge a_j D_j \right).
\end{aligned}
\end{equation*}
Then by the Borell-TIS inequality again, for any fixed $\ep_0>0$ and sufficiently large $a_j$,
\begin{equation}
\label{eq:I-mode-bound-1}
\P\left(\# \{t \in \tilde{T}\cap I_j^{\rm mode})\} \ge 2\right) \le \exp\left(-\frac{a_j^2D_j^2}{2F+\ep_0} \right).
\end{equation}
Putting \eqref{eq:I-mode-bound-0} and \eqref{eq:I-mode-bound-1} together gives the bound in the second line of \eqref{eq:I-bound}.

(3) The probability that no local maxima of $y_\gamma(t)$ in $I_j^{\rm mode}$ exceed the threshold $u$ is less than the probability that $y_\gamma(t)$ is below $u$ anywhere in $I_j^{\rm mode}$, so it is bounded above by $\Phi[(u-a_j M_j)/\sigma_\gamma]$. On the other hand, the probability that more than one local maxima of $y_\gamma(t)$ in $I_j^{\rm mode}$ exceed $u$ is less than the probability that there exist more than one local maximum, which is bounded above by \eqref{eq:I-mode-bound-1}. Putting the two together gives the bound in the last line of \eqref{eq:I-bound}.
\end{proof}

\begin{remark}
Recall the uniformity assumptions in our model in Section \ref{sec:model}, especially $\sup_j |S_{j, \gamma}| < \infty$ and $\inf_j M_j >0$. It can be seen from the proof of Lemma \ref{lemma:bounds} that there exist a universal $K>0$ such that \eqref{eq:I-bound} holds for all $j$ with $a_j>K$.
\end{remark}

\begin{lemma}
\label{lemma:unique-max}
Assume the model of Section \ref{sec:model} and let $\tilde{m}_{1, \gamma} = \#\{ \tilde{T}\cap \mathbb{S}_{1, \gamma} \}$ and $\tilde{m}_{1, \gamma}(u) = \#\{t\in \tilde{T}(u)\cap \mathbb{S}_{1, \gamma} \}$. Under conditions (C1) and (C2), there exists some $\ep_1>0$ such that as $L\to \infty$,
\begin{equation}\label{Eq:rate-localmax}
\begin{split}
\P\left(\# \{t\in \tilde{T}\cap \mathbb{T}_\gamma \}\ge 1\right)  & =o(e^{-\ep_1 a^2}),\\
\P\left( \tilde{m}_{1, \gamma} = J\right)  &=1-o(e^{-\ep_1 a^2}),\\
\P\left( \tilde{m}_{1, \gamma}(u) = J\right) &=1-o(e^{-\ep_1 a^2}),
\end{split}
\end{equation}
where $o(e^{-\ep_1 a^2})=o(L^{-N})$ by condition (C2).
\end{lemma}
\begin{proof} (1) Write $\mathbb{T}_\gamma = \cup_{j=1}^{J} T_{j,\gamma}$, where $T_{j,\gamma} = S_{j,\gamma} \setminus S_j$ is the transition region for peak $j$. Note that $T_{j,\gamma}$ is a subset of $I_j^{\rm side}$ since $I_j^{\rm mode} \subset S_j$. By \eqref{eq:I-bound} and condition (C2), for sufficiently large $L$, the required probability $\P(\# \{t\in \tilde{T}\cap \mathbb{T}_\gamma \} \ge 1 )$ is bounded above by
\[
\sum_{j=1}^{J} \exp\left(-\frac{a_j^2C_j^2}{2E+\ep_0} \right)
\le
L^N \exp\left(-\frac{a^2C^2}{2E+\ep_0} \right).
\]
Now the first line of \eqref{Eq:rate-localmax} follows from the fact that for any $\delta>0$
\[
L^N \exp(-\delta a^2) = \exp\left\{ a^2\left( \frac{N \log L}{a^2} - \delta\right)\right\} \to 0.
\]

(2) By \eqref{eq:I-bound} and condition (C2), for sufficiently large $J$, the required probability is bounded below by
\begin{equation*}
\begin{aligned}
& \P \left[\cap_{j=1}^{J} \left(\#\{t \in \tilde{T}\cap I_j^{\rm mode} \} = 1
~\cap~ \#\{t \in \tilde{T}\cap I_j^{\rm side} \} = 0 \right) \right] \\
&\ge 1- \sum_{j=1}^{J} \left[1 - \P\left(\#\{t \in \tilde{T}\cap I_j^{\rm mode} \} = 1 ~\cap~ \#\{t \in \tilde{T}\cap I_j^{\rm side} \} = 0 \right) \right] \\
&\ge 1- \sum_{j=1}^{J} \left[
2\exp\left(-\frac{a_j^2C_j^2}{2E+\ep_0} \right) + \exp\left(-\frac{a_j^2D_j^2}{2F+\ep_0} \right) \right] \\
&\ge
1 - 2L^N \exp\left(-\frac{a^2C^2}{2E+\ep_0} \right)- L^N \exp\left(-\frac{a^2D^2}{2F+\ep_0} \right),
\end{aligned}
\end{equation*}
where $D > 0$ is the infimum of the $D_i$'s. The second line of \eqref{Eq:rate-localmax} follows.

(3) Applying the last line of \eqref{eq:I-bound}, together with similar argument in step (2), yields the last line of \eqref{Eq:rate-localmax}.
\end{proof}

\begin{lemma}
\label{lemma:variance-localmax}
Assume the model of Section \ref{sec:model} and let $\tilde{m}(v)$ be as defined in Section \ref{sec:alg}. Then ${\rm Var}(\tilde{m}(v)) = O(L^N)$ as $L \to \infty$.
\end{lemma}

\begin{proof}
By the Kac-Rice formula \citep{Adler:2007}, $\E[\tilde{m}(v)]$ equals
\begin{equation*}
\int_{U(L)} \E\{|{\rm det}\nabla^2 y_\gamma(t)| \mathbbm{1}_{\{y_\gamma(t)>v, \nabla^2 y_\gamma(t)\prec 0\}}|\nabla y_\gamma(t)=0\}p_{\nabla y_\gamma(t)}(0)dt
\end{equation*}
and $\E[\tilde{m}(v)(\tilde{m}(v)-1)]$ equals
\begin{equation*}
\begin{split}
&\int_{U(L)} \int_{U(L)} \E\{|{\rm det}\nabla^2 y_\gamma(t)| |{\rm det}\nabla^2 y_\gamma(s)|\mathbbm{1}_{\{y_\gamma(t)>v, \nabla^2 y_\gamma(t)\prec 0\}}\\
& \quad \times \mathbbm{1}_{\{y_\gamma(s)>v, \nabla^2 y_\gamma(s)\prec 0\}}|\nabla y_\gamma(t)=0, \nabla y_\gamma(s)=0\} p_{\nabla y_\gamma(t),\nabla y_\gamma(s)}(0,0)dtds.
\end{split}
\end{equation*}
It then follows from the observation ${\rm Var}(\tilde{m})=\E[\tilde{m}(\tilde{m}-1)] + \E[\tilde{m}] - (\E[\tilde{m}])^2$
that ${\rm Var}(\tilde{m})$ can be written as
\begin{equation*}
\begin{split}
&\int_{U(L)} \int_{U(L)} \Big( \E\{|{\rm det}\nabla^2 y_\gamma(t)| |{\rm det}\nabla^2 y_\gamma(s)|\mathbbm{1}_{\{y_\gamma(t)>v,\nabla^2 y_\gamma(t)\prec 0\}}\\
&\quad \times \mathbbm{1}_{\{y_\gamma(s)>v, \nabla^2 y_\gamma(s)\prec 0\}} |\nabla y_\gamma(t)=0, \nabla y_\gamma(s)=0\} p_{\nabla y_\gamma(t),\nabla y_\gamma(s)}(0,0)\\
&\quad -\E\{|{\rm det}\nabla^2 y_\gamma(t)| \mathbbm{1}_{\{y_\gamma(t)>v, \nabla^2 y_\gamma(t)\prec 0\}}|\nabla y_\gamma(t)=0\} p_{\nabla y_\gamma(t)}(0) \\
&\quad \times \E\{|{\rm det}\nabla^2 y_\gamma(s)| \mathbbm{1}_{\{y_\gamma(s)>v, \nabla^2 y_\gamma(s)\prec 0\}}|\nabla y_\gamma(s)=0\}p_{\nabla y_\gamma(s)}(0)\Big)dtds\\
&\quad + \int_{U(L)} \E\{|{\rm det}\nabla^2 y_\gamma(t)| \mathbbm{1}_{\{y_\gamma(t)>v, \nabla^2 y_\gamma(t)\prec 0\}}|\nabla y_\gamma(t)=0\}p_{\nabla y_\gamma(t)}(0)dt\\
&:= \int_{U(L)} \int_{U(L)} f_1(t,s)dtds + \int_{U(L)} f_2(t) dt :=  I_1+I_2.
\end{split}
\end{equation*}
Recall $y_\gamma(t)=z_\gamma(t)+\mu_\gamma(t)$, where $\mu_\gamma(t)=\sum_{j=-\infty}^\infty  a_jh_{j,\gamma}(t)$ and $z_\gamma(t)$ is stationary. Assume $t\in S_{j, \gamma}$ for some $j$, then
\begin{equation}\label{eq:f2(t)}
\begin{split}
f_2(t) &\le \E\{|{\rm det}(\nabla^2 z_\gamma(0)+ a_j \nabla^2 h_{j,\gamma}(t)) | |\nabla z_\gamma(0) + a_j \nabla h_{j,\gamma}(t)=0\}\\
&\quad \times \frac{\exp\left[-a_j^2[ \nabla h_{j,\gamma}(t)]^T[{\rm Cov}(\nabla z_\gamma(0))]^{-1} \nabla h_{j,\gamma}(t)/2\right]}{(2\pi)^{N/2}[{\rm det Cov}(\nabla z_\gamma(0))]^{1/2}}\\
&\le c_1(a_j^{n_1}+1)\exp[-c_2a_j^2]\le c_3
\end{split}
\end{equation}
for some positive constants $n_1$, $c_1$, $c_2$ and $c_3$. Similar result holds for $t$ such that $\mu_\gamma(t)=0$. Therefore $I_2=O(L^N)$.

We consider a pair of independent Gaussian fields $\tilde{z}_\gamma(t)$ and $\tilde{z}_\gamma(s)$ with the distribution of $z_\gamma(t)$ and $z_\gamma(s)$ respectively, assume also that $\tilde{z}_\gamma(t)$ and $\tilde{z}_\gamma(s)$ are both independent of $z_\gamma(t)$ and $z_\gamma(s)$. Let
\begin{equation*}
\begin{split}
Z_1(t,s) &=(z_\gamma(t), z_\gamma(s)),  \quad Y_1(t,s) = Z_1(t,s) + (\mu_\gamma(t), \mu_\gamma(s)),\\
Z_0(t,s) &=(\tilde{z}_\gamma(t), \tilde{z}_\gamma(s)),  \quad Y_0(t,s) = Z_0(t,s) + (\mu_\gamma(t), \mu_\gamma(s)),\\
\end{split}
\end{equation*}
and let $F(x,y,A, B)=|{\rm det}(A)| |{\rm det}(B)|\mathbbm{1}_{\{x>v, A\prec 0\}}\mathbbm{1}_{\{y>v, B\prec 0\}}$, where $x,y\in \R$ and $A$ and $B$ are $N\times N$ symmetric matrices. Then
\begin{equation*}
\begin{split}
f_1(t,s) &= \E\{F(Y_1(t,s), \nabla^2 Y_1(t,s))|\nabla Y_1(t,s)=0\}p_{\nabla Y_1(t,s)}(0)\\
&\quad - \E\{F(Y_0(t,s), \nabla^2 Y_0(t,s))|\nabla Y_0(t,s)=0\}p_{\nabla Y_0(t,s)}(0).
\end{split}
\end{equation*}
Let
\begin{equation*}
\begin{split}
Y_\eta(t,s)&=\sqrt{\eta}Z_1(t,s) + \sqrt{1-\eta}Z_0(t,s) + (\mu_\gamma(t), \mu_\gamma(s)),\\
Y_\eta'(t,s)&=\sqrt{\eta}\nabla Z_1(t,s) + \sqrt{1-\eta}\nabla Z_0(t,s) + (\nabla \mu_\gamma(t), \nabla \mu_\gamma(s)),\\
Y_\eta^{''}(t,s)&=\sqrt{\eta}\nabla^2 Z_1(t,s) + \sqrt{1-\eta}\nabla^2 Z_0(t,s) + (\nabla^2 \mu_\gamma(t), \nabla^2 \mu_\gamma(s)),
\end{split}
\end{equation*}
and denote by $R_\eta(i,j)$ the $(i,j)$ entry of ${\rm Cov}(Y_\eta(t,s), Y_\eta'(t,s),Y_\eta^{''}(t,s))$. It follows that $f_1(t,s)$ can be written as
\begin{equation*}
\begin{split}
 &\quad \int_0^1 \frac{d}{d\eta}\left(\E\{F(Y_\eta(t,s),Y_\eta^{''}(t,s))|Y_\eta'(t,s)=0\}p_{Y_\eta'(t,s)}(0)\right)d\eta \\
&=\int_0^1 \frac{d}{d\eta}\int_{\R^{2+N(N+1)}} F(\xi)p_{(Y_\eta(t,s),Y_\eta^{''}(t,s)), Y_\eta'(t,s)}(\xi, 0)d\xi d\eta\\
&=\int_0^1 \int_{\R^{2+N(N+1)}} F(\xi)\frac{d}{d\eta} p_{(Y_\eta(t,s),Y_\eta^{''}(t,s)), Y_\eta'(t,s)}(\xi, 0)d\xi d\eta\\
&=\int_0^1 \int_{\R^{2+N(N+1)}} F(\xi)\left(\sum_{i,j}\frac{\partial p_{(Y_\eta(t,s),Y_\eta^{''}(t,s)), Y_\eta'(t,s)}(\xi, 0)}{\partial R_\eta(i,j)} \frac{\partial R_\eta(i,j)}{\partial \eta} \right)d\xi d\eta.
\end{split}
\end{equation*}
By a well-known formula for Gaussian density which is similar to the heat equation \citep[Eq. (2.2.6)]{Adler:2007} and the fact that $R_\eta(i,j)=\eta R_1(i,j)+(1-\eta)R_0(i,j)$, we see that
\begin{equation*}
\begin{split}
|f_1(t,s)| &\le \sum_{i,j} \left|\int_0^1 \int_{\R^{2+N(N+1)}} F(\xi) \frac{\partial^2 p_{(Y_\eta(t,s),Y_\eta^{''}(t,s)), Y_\eta'(t,s)}(y)}{\partial y_i \partial y_j}\Big|_{y=(\xi, 0)} d\xi d\eta\right|\\
&\qquad \quad \times |R_1(i,j)-R_0(i,j)|.
\end{split}
\end{equation*}
Similarly to \eqref{eq:f2(t)}, there exists a positive constant $c_4$ such that
\begin{equation*}
\begin{split}
\left|\int_0^1 \int_{\R^{2+N(N+1)}} F(\xi) \frac{\partial^2 p_{(Y_\eta(t,s),Y_\eta^{''}(t,s)), Y_\eta'(t,s)}(y)}{\partial y_i \partial y_j}\Big|_{y=(\xi, 0)} d\xi d\eta\right|\le c_4
\end{split}
\end{equation*}
for all $i$ and $j$.

Recall $r_\gamma(t)=\E[z_\gamma(t)z_\gamma(0)]$ in Section \ref{sec:model}. By stationarity and change of variables, we obtain
\begin{equation*}
\begin{split}
|I_1| & \le 8c_4 2^N L^N\sum_{i,j,k,l=1}^N\int_0^L \cdots\int_0^L \left(\prod_{i=1}^N\bigg(1-\frac{t_i}{L}\bigg)\right)\Bigg(|r_\gamma(t)| + \left|\frac{\partial r_\gamma(t)}{\partial t_i} \right|\\
&\quad +\left|\frac{\partial^2 r_\gamma(t)}{\partial t_i \partial t_j} \right| + \left|\frac{\partial^3 r_\gamma(t)}{\partial t_i \partial t_j \partial t_k} \right| + \left|\frac{\partial^4 r_\gamma(t)}{\partial t_i \partial t_j\partial t_k \partial t_l} \right| \Bigg) dt_1\cdots dt_N.
\end{split}
\end{equation*}
Condition \eqref{Eq:variance-asymptotics} implies that the last integral above is finite as $L \to \infty$, hence $I_1=O(L^N)$, completing the proof.
\end{proof}


\begin{lemma}
\label{lemma:mean-m-1}
Assume the model of Section \ref{sec:model} and let $\tilde{m}_{1, \gamma}$ and $\tilde{m}_{1, \gamma}(u)$ be as defined in Lemma \ref{lemma:unique-max}. Then
\begin{equation*}
\E[\tilde{m}_{1, \gamma}]/J = 1 + O(a^{-2}), \quad  \E[\tilde{m}_{1, \gamma}(u)]/J = 1 + O(a^{-2}).
\end{equation*}
\end{lemma}

\begin{proof}
We only prove the first part since the second part can be derived similarly. By the Kac-Rice formula \citep{Adler:2007}, $\E[\tilde{m}_{1, \gamma}]=\sum_{j=1}^J I(S_{j,\gamma})$, where
\begin{equation*}
\begin{split}
I(S_{j,\gamma}) &= \int_{S_{j,\gamma}} \E\left\{|{\rm det}\nabla^2 y_\gamma(t)| \mathbbm{1}_{\{\nabla^2 y_\gamma(t)\prec 0\}}\big|\nabla y_\gamma(t)=0\right\}p_{\nabla y_\gamma(t)}(0)dt\\
&= \frac{1}{(2\pi)^{N/2}\sqrt{{\rm det }(\La_\gamma)}}\int_{S_{j,\gamma}} \E\big\{|{\rm det}(\nabla^2 z_\gamma(0)+ a_j \nabla^2 h_{j,\gamma}(t))| \\
&\quad \times \mathbbm{1}_{\{\nabla^2 z_\gamma(0)+ a_j \nabla^2 h_{j,\gamma}(t) \prec 0\}}\big| \nabla z_\gamma(0) + a_j \nabla h_{j,\gamma}(t)=0\big\}\\
&\quad \times \exp\left\{-a_j^2[ \nabla h_{j,\gamma}(t)]^T\La_\gamma^{-1} \nabla h_{j,\gamma}(t)/2\right\}dt.
\end{split}
\end{equation*}
Let $f(t)= [ \nabla h_{j,\gamma}(t)]^T\La_\gamma^{-1} \nabla h_{j,\gamma}(t)/2$, which attains the minimum 0 only at $\tau_j$ over $S_{j,\gamma}$ such that
\begin{equation}\label{eq:hessian}
\nabla^2 f(\tau_j)= \nabla^2 h_{j,\gamma}(\tau_j)\La_\gamma^{-1} \nabla^2 h_{j,\gamma}(\tau_j).
\end{equation}
Let $I_j^{\rm mode}$, $C$ and $D$ be as defined in Section \ref{sec:model}. Since $C>0$, $f(t)$ is strictly greater than 0 over $S_{j,\gamma}\backslash I_j^{\rm mode}$ , and there exists some $\ep_0>0$ such that $I(S_{j,\gamma}) = I(I_j^{\rm mode}) + o(e^{-\ep_0 a^2})$ for all $j$. Moreover, since $D>0$, a similar argument as in the proof of \citep[Theorem 2.4]{CS:2014} yields that removing the indicator function $\mathbbm{1}_{\{\nabla^2 z_\gamma(0)+ a_j \nabla^2 h_{j,\gamma}(t) \prec 0\}}$ in the expression of $I(I_j^{\rm mode})$ only causes an error of rate $o(e^{-\ep_1 a^2})$, where $\ep_1>0$ is some constant. Now, applying the Laplace method to the expression
\begin{equation*}
\begin{split}
&\frac{1}{(2\pi)^{N/2}\sqrt{{\rm det }(\La_\gamma)}}\int_{I_j^{\rm mode}} \E\big\{|{\rm det}(\nabla^2 z_\gamma(0)+ a_j \nabla^2 h_{j,\gamma}(t))| \\
&\quad \times \mathbbm{1}_{\{\nabla^2 z_\gamma(0)+ a_j \nabla^2 h_{j,\gamma}(t) \prec 0\}}\big| \nabla z_\gamma(0) + a_j \nabla h_{j,\gamma}(t)=0\big\} \exp\left\{-a_j^2f(t)\right\}dt,
\end{split}
\end{equation*}
we obtain
\begin{equation*}
\begin{split}
I(S_{j,\gamma}) &= \frac{ {\rm det}(a_j \nabla^2 h_{j,\gamma}(\tau_j))}{(2\pi)^{N/2}\sqrt{{\rm det }(\La_\gamma)}} \left( \frac{(2\pi)^N}{a_j^{2N} {\rm det}(\nabla^2 f(\tau_j))}\right)^{1/2} + O(a_j^{-2}) \\
&= 1 + O(a_j^{-2}),
\end{split}
\end{equation*}
where the last line is due to \eqref{eq:hessian}. Therefore $\E[\tilde{m}_{1, \gamma}]/J = 1 + O(a^{-2})$.
\end{proof}

\subsection{Control of FDR}
\label{app:FDR}

\begin{proof}[Proof of Theorem \ref{thm:FDR}]
We prove part (ii) first, since part (i) is much easier and can be seen from the proof of part (ii). Let $\tilde{G}(u,v) = \#\{t\in\tilde{T}(u)\}/\#\{t\in\tilde{T}(v)\}$ be the empirical marginal right cdf of $y_\gamma(t)$ given $t \in \tilde{T}(v)$, where $u>v$. Then the random BH threshold $\tilde{u}_{\BH}(v)$ \eqref{eq:thresh-BH-random} satisfies
\[
\alpha \tilde{G}(\tilde{u}_{\BH}(v),v) = k\alpha/\tilde{m}(v) = F_\gamma(\tilde{u}_{\BH}(v),v),
\]
so $\tilde{u}_{\BH}(v)$ is the smallest $u$ that is greater than $v$ and satisfies
\begin{equation}
\label{eq:FDRthreshold}
\alpha \tilde{G}(u,v) \ge F_\gamma(u,v).
\end{equation}
The strategy is to solve equation \eqref{eq:FDRthreshold} in the limit when $L, a \to \infty$. We first find the limit of $\tilde{G}(u,v)$. Letting $\tilde{m}_{0, \gamma}(u) = \#\{t \in \tilde{T}(u) \cap\mathbb{S}_{0,\gamma}\}$ and $\tilde{m}_{1, \gamma}(u) = \#\{t \in \tilde{T}(u) \cap\mathbb{S}_{1,\gamma}\}$, so that $\tilde{m}(u) = \tilde{m}_{0, \gamma}(u)+ \tilde{m}_{1, \gamma}(u)$, write
\begin{equation}
\label{eq:ecdf}
\tilde{G}(u,v) = \frac{\tilde{m}(u)}{\tilde{m}(v)} =
\frac{\tilde{m}_{0, \gamma}(u)}{\tilde{m}_{0, \gamma}(v) +\tilde{m}_{1,\gamma}(v)}
+ \frac{\tilde{m}_{1, \gamma}(u)}{\tilde{m}_{0, \gamma}(v) +\tilde{m}_{1,\gamma}(v)}.
\end{equation}
By Chebyshev's inequality, Lemma \ref{lemma:variance-localmax} and condition (C2),
\begin{equation*}
\begin{split}
\tilde{m}_{0, \gamma}(v)/L^N &= \E[\tilde{m}_{0, \gamma}(v)]/L^N + O_p(L^{-N/2})\\
&=\E[\tilde{m}_{0, \gamma}(U(1),v)](1-A_{2,\gamma}) +O_p(a^{-2}+L^{-N/2}).
\end{split}
\end{equation*}
On the other hand, by Lemma \ref{lemma:mean-m-1} and condition (C2),
\begin{equation*}
\begin{split}
\tilde{m}_{1, \gamma}(v)/L^N = \E[\tilde{m}_{1, \gamma}(v)]/L^N +O_p(L^{-N/2})=A_1 +O_p(a^{-2}+L^{-N/2}).
\end{split}
\end{equation*}
Therefore, \eqref{eq:ecdf} can be written as
\begin{equation*}
\begin{split}
\tilde{G}(u,v) &= \frac{\E[\tilde{m}_{0, \gamma}(U(1),u)](1-A_{2,\gamma}) + A_1}{\E[\tilde{m}_{0, \gamma}(U(1),v)](1-A_{2,\gamma}) + A_1} + O_p(a^{-2}+L^{-N/2})\\
&=  \frac{F_\gamma(u,v)\E[\tilde{m}_{0, \gamma}(U(1),v)](1-A_{2,\gamma}) + A_1}{\E[\tilde{m}_{0, \gamma}(U(1),v)](1-A_{2,\gamma}) + A_1} + O_p(a^{-2}+L^{-N/2}).
\end{split}
\end{equation*}
Now replacing $\tilde{G}(u,v)$ by its limit in \eqref{eq:FDRthreshold}, and solving for $u$ gives the deterministic solution
\begin{equation}
\label{eq:u_BH^*}
u^*_{\BH}(v) =F_\gamma(\cdot, v)^{-1} \left(\frac {\alpha A_1}{A_1 + \E[\tilde{m}_{0, \gamma}(U(1),v)](1-A_{2,\gamma})(1-\alpha)}\right).
\end{equation}
The argument above implies $F(\tilde{u}_{\BH}(v),v) = F(u^*_{\BH}(v),v) + O_p(a^{-2}+L^{-N/2})$, but the first-derivative of $F_\gamma(\cdot, v)$ is uniformly bounded, thus $\tilde{u}_{\BH}(v) = u^*_{\BH}(v) + O_p(a^{-2}+L^{-N/2})$. Following similar arguments, together with Chebyshev's inequality, we also have
\begin{equation}
\label{eq:u_BH^*-and-utilde}
\begin{split}
\P\left(|\tilde{u}_{\BH}(v) - u^*_{\BH}(v)| > a^{-1} + L^{-N/4}\right) = O(a^{-1} + L^{-N/4}).
\end{split}
\end{equation}

Next, let us turn to estimating $\E[\tilde{m}_{0, \gamma}(\tilde{u}_{\BH}(v))]$. Notice that,
\begin{equation*}
\begin{split}
&\quad \left|\E\left[(\tilde{m}_{0, \gamma}(\tilde{u}_{\BH}(v))- \tilde{m}_{0, \gamma}(u^*_{\BH}(v)))/L^N\right]\right|\\
&\le \left|\E\left[(\tilde{m}_{0, \gamma}(\tilde{u}_{\BH}(v))- \tilde{m}_{0, \gamma}(u^*_{\BH}(v)))/L^N\mathbbm{1}_{\{|\tilde{u}_{\BH}(v) - u^*_{\BH}(v)| \le a^{-1} + L^{-N/4}\}}\right]\right|\\
&\quad + \left|\E\left[(\tilde{m}_{0, \gamma}(\tilde{u}_{\BH}(v))- \tilde{m}_{0, \gamma}(u^*_{\BH}(v)))/L^N\mathbbm{1}_{\{|\tilde{u}_{\BH}(v) - u^*_{\BH}(v)| > a^{-1} + L^{-N/4}\}}\right]\right|\\
&\le \left|\E\left[\left(\tilde{m}_{0, \gamma}\left(u^*_{\BH}(v)-a^{-1} - L^{-N/4}\right)-\tilde{m}_{0, \gamma}(u^*_{\BH}(v))\right)/L^N\right]\right| \\
&\quad + \left|\E\left[\left(\tilde{m}_{0, \gamma}(u^*_{\BH}(v)) - \tilde{m}_{0, \gamma}\left(u^*_{\BH}(v)+a^{-1} + L^{-N/4}\right)\right)/L^N\right]\right|\\
&\quad + 2\left(\E\left[\tilde{m}(v)^2\right]/L^{2N}\right)^{1/2}\P\left(|\tilde{u}_{\BH}(v) - u^*_{\BH}(v)| > a^{-1} + L^{-N/4}\right)\\
&= O(a^{-1} + L^{-N/4}),
\end{split}
\end{equation*}
where the second inequality is due to H\"older's inequality and the monotonicity of $\tilde{m}_{0, \gamma}(\cdot)$, and the last line is derived by applying Taylor's expansion to $F_\gamma(\cdot)$, together with Lemma \ref{lemma:variance-localmax} and \eqref{eq:u_BH^*-and-utilde}. Hence
\begin{equation}\label{eq:Vgamma-tu}
\begin{split}
\E[\tilde{m}_{0, \gamma}(\tilde{u}_{\BH}(v))/L^N] &= \E[\tilde{m}_{0, \gamma}(u^*_{\BH}(v))/L^N] + O(a^{-1} + L^{-N/4})\\
&= \E[\tilde{m}_{0, \gamma}(U(1),u^*_{\BH}(v))](1-A_{2,\gamma}) + O(a^{-1}+ L^{-N/4}).
\end{split}
\end{equation}
A similar argument, together with Lemma \ref{lemma:unique-max} and \eqref{eq:u_BH^*-and-utilde}, yields
\begin{equation}\label{eq:rand-neq-J}
\P(\tilde{m}_{1, \gamma}(\tilde{u}_{\BH}(v))\ne J) = O(a^{-1} + L^{-N/4}).
\end{equation}

Let $W(u)=R(u)-V(u)$. By \citep[Lemma 12]{Schwartzman:2011}, Lemma \ref{lemma:unique-max}, condition (C2), \eqref{eq:Vgamma-tu} and \eqref{eq:rand-neq-J}, we obtain that $\FDR_{\BH}(v)$ is bounded above by
\begin{equation}
\label{eq:FDR-u*}
\begin{aligned}
\P&(W(\tilde{u}_{\BH}(v))\le J-1) + \frac{ \E\left[ V(\tilde{u}_{\BH}(v))\right]}{ \E\left[V(\tilde{u}_{\BH}(v))\right] + J } \\
&= \P(\tilde{m}_{1, \gamma}(\tilde{u}_{\BH}(v))\le J-1) + \frac{\E\left[ \tilde{m}_{0, \gamma}(\tilde{u}_{\BH}(v))/L^N\right] }{ \E\left[\tilde{m}_{0, \gamma}(\tilde{u}_{\BH}(v))/L^N\right] + J/L^N } + o(L^{-N})\\
&= \frac{\E[\tilde{m}_{0, \gamma}(U(1),u^*_{\BH}(v))](1-A_{2,\gamma}) }{ \E[\tilde{m}_{0, \gamma}(U(1),u^*_{\BH}(v))](1-A_{2,\gamma}) + A_1} + O(a^{-1} + L^{-N/4})\\
&= \frac{F_\gamma(u^*_{\BH}(v),v)\E[\tilde{m}_{0, \gamma}(U(1),v)](1-A_{2,\gamma}) }{ F_\gamma(u^*_{\BH}(v),v)\E[\tilde{m}_{0, \gamma}(U(1),v)](1-A_{2,\gamma}) + A_1}  + O(a^{-1} + L^{-N/4}),
\end{aligned}
\end{equation}
where we have split $\tilde{m}_{1, \gamma}(\tilde{u}_{\BH}(v))$ into the true signal region $\mathbb{S}_{1}$ and the transition region $\mathbb{T}_\gamma = \mathbb{S}_{1,\gamma} \setminus \mathbb{S}_{1}$ and split $V(\tilde{u}_{\BH}(v))$ into the reduced null region $\mathbb{S}_{0, \gamma}$ and the transition region $\mathbb{T}_\gamma = \mathbb{S}_0 \setminus \mathbb{S}_{0, \gamma}$ as well. Plugging \eqref{eq:u_BH^*} into the last line of \eqref{eq:FDR-u*} yields part (ii).

For a deterministic threshold $u$, a similar argument for showing \eqref{eq:FDR-u*} yields part (i).
\end{proof}

\subsection{Power}
\label{app:power}


\begin{proof}[Proof of Lemma \ref{lemma:power-approx}] Recall $M_j = h_{j,\gamma}(\tau_j)$. By Lemma \ref{lemma:bounds},
\begin{equation*}
\begin{split}
\Power_j(u) &= \P\left(\#\{t \in \tilde{T}(u)\cap S_j\} \ge 1\right)\ge \P\left(\#\{t \in \tilde{T}(u)\cap I_j^{\rm mode}\} = 1 \right)\\
&\ge \Phi\left(\frac{a_jM_j-u}{\sigma_\gamma}\right) - \exp\left(-\frac{a_j^2D_j^2}{2F+\ep_0} \right).
\end{split}
\end{equation*}
On the other hand, by Lemma \ref{lemma:bounds} again,
\begin{equation*}
\begin{split}
\Power_j(u) &\le \P\left(\#\{t \in \tilde{T}(u)\cap I_j^{\rm mode}\} \ge 1 \right) + \P\left(\#\{t \in \tilde{T}(u)\cap I_j^{\rm side}\} \ge 1 \right)\\
&\le \E\left(\#\{t \in \tilde{T}(u)\cap I_j^{\rm mode}\} \right) + \exp\left(-\frac{a_j^2C_j^2}{2E+\ep_0} \right).
\end{split}
\end{equation*}
Similarly to the proof of Lemma \ref{lemma:mean-m-1}, applying the Laplace method, we obtain
\begin{equation*}
\begin{split}
\E\left(\#\{t \in \tilde{T}(u)\cap I_j^{\rm mode}\} \right) = \Phi\left(\frac{a_jM_j-u}{\sigma_\gamma}\right)(1 + O(a_j^{-2})),
\end{split}
\end{equation*}
completing the proof.
\end{proof}

\begin{proof}[Proof of Theorem \ref{thm:power}]
Part (i) follows immediately from Lemma \ref{lemma:power-approx}. For part (ii), notice that for each $j$ and any fixed $\delta>0$,
\begin{equation*}
\begin{split}
& \quad \P \left(\#\{t \in \tilde{T}(\tilde{u}_{\BH}(v))\cap S_j\} \ge 1\right)\\
& = \P \left(\#\{t \in \tilde{T}(\tilde{u}_{\BH}(v))\cap S_j\} \ge 1, \ |\tilde{u}_{\BH}(v) - u^*_{\BH}(v)| \le \delta \right)\\
& \quad +  \P \left(\#\{t \in \tilde{T}(\tilde{u}_{\BH}(v))\cap S_j\} \ge 1, \ |\tilde{u}_{\BH}(v) - u^*_{\BH}(v)| > \delta \right)\\
& \le \P \left(\#\{t \in \tilde{T}(u^*_{\BH}(v) - \delta)\cap S_j\} \ge 1 \right) + \P \left(|\tilde{u}_{\BH}(v) - u^*_{\BH}(v)| > \delta \right),
\end{split}
\end{equation*}
then the result follows from Lemma \ref{lemma:power-approx} and the fact similar to \eqref{eq:u_BH^*-and-utilde} that $\P \left(|\tilde{u}_{\BH}(v) - u^*_{\BH}(v)| > \delta \right) = O(a^{-2} + L^{-N/2})$.
\end{proof}


\subsection{Approximating the overshoot distribution}
\label{app:approx-os}

\begin{proof}[Proof of Theorem \ref{thm:FDR-os-approx}]
If we use $K_\gamma(\cdot,v)$ instead of $F_\gamma(\cdot,v)$ to compute the p-values, then the random BH threshold $\tilde{u}_{\BH}(v)$ \eqref{eq:thresh-BH-random} satisfies $\alpha \tilde{G}(\tilde{u}_{\BH}(v),v) = K_\gamma(\tilde{u}_{\BH},v)$, so $\tilde{u}_{\BH}(v)$ is the smallest $u$ that is greater than $v$ and satisfies
\begin{equation}
\label{eq:FDRthreshold-os-approx}
\alpha \tilde{G}(u,v) \ge K_\gamma(u,v).
\end{equation}
A similar argument in the proof of Theorem \ref{thm:FDR} gives
\begin{equation*}
\tilde{G}(u,v) =  \frac{F_\gamma(u,v)\E[\tilde{m}_{0, \gamma}(U(1),v)](1-A_{2,\gamma}) + A_1}{\E[\tilde{m}_{0, \gamma}(U(1),v)](1-A_{2,\gamma}) + A_1} + O_p(a^{-2}+L^{-N/2}).
\end{equation*}
By \eqref{Eq:stationary-os-u}, $K_\gamma(u,v)=F_\gamma(u,v)/\beta_\gamma(v) (1+ o(e^{-\ep_0 v^2}))$, replacing $\tilde{G}(u,v)$ by its limit in \eqref{eq:FDRthreshold-os-approx} and then solving for $u$ gives the deterministic solution $u^{**}_{\BH}(v)$ such that
\begin{equation}
\label{eq:u_BH^*-os-approx}
F_\gamma(u^{**}_{\BH}(v),v) = \frac {\alpha \beta_\gamma(v) A_1(1+ o(e^{-\ep_0 v^2}))}{A_1 + \E[\tilde{m}_{0, \gamma}(U(1),v)](1-A_{2,\gamma})(1-\alpha\beta_\gamma(v))}.
\end{equation}
By conditions $v^2/\log(L) \to 0$ and  $v^2/\log(a) \to 0$, $a^{-1}+L^{-N/4}=o(e^{-c v^2})$ for any fixed $c>0$ and moreover, Lemma \ref{lemma:variance-localmax} and Lemma \ref{lemma:mean-m-1} still hold as $v\to \infty$. Therefore, similarly to \eqref{eq:FDR-u*}, we obtain
\begin{equation*}
\begin{aligned}
\FDR_{\BH}(v) &\le \frac{F_\gamma(u^{**}_{\BH}(v),v)\E[\tilde{m}_{0, \gamma}(U(1),v)](1-A_{2,\gamma}) }{ F_\gamma(u^{**}_{\BH}(v),v)\E[\tilde{m}_{0, \gamma}(U(1),v)](1-A_{2,\gamma}) + A_1}\\
&\quad + O(a^{-1} + L^{-N/4})\\
& = \alpha\frac{\E[\tilde{m}_{0, \gamma}(U(1),v)](1-A_{2,\gamma})\beta_\gamma(v)}{\E[\tilde{m}_{0, \gamma}(U(1),v)](1-A_{2,\gamma}) + A_1} (1+ o(e^{-\ep_0 v^2})).
\end{aligned}
\end{equation*}
This proves \eqref{eq:bound-FDR-os-approx}.

On the other hand, we can slightly modify the proofs of Lemma \ref{lemma:power-approx} and Theorem \ref{thm:power} to obtain \eqref{eq:power-os-approx}.
\end{proof}


\bibliographystyle{plainnat}

\end{document}